\documentclass[journal]{IEEEtran}
\usepackage{amsmath,graphicx}
\usepackage{amssymb}  
\usepackage[utf8]{inputenc}
\usepackage[amsmath,thmmarks,hyperref]{ntheorem} 
\usepackage{bbm}               
\usepackage{mathtools}         
\usepackage{stmaryrd}          
\usepackage{paralist}          
\usepackage{aliascnt}          
\usepackage{subfig} 
\usepackage{mathrsfs}
\usepackage{hyperref}
\usepackage{textcomp}
\usepackage{color}             
\usepackage{booktabs}          
\usepackage[export]{adjustbox} 
\usepackage[normalem]{ulem} 

\usepackage{algorithm}
\usepackage{algpseudocode}

\graphicspath{{pics/}}

\newcommand\MTkillspecial[1]{
\bgroup
\catcode`\&=9
\let\\\relax%
\scantokens{#1}
\egroup
}

\renewcommand{\mathbb}[1]{\mathbbm{#1}}

\renewcommand{\theenumi}{\roman{enumi}}

\theoremheaderfont{\normalfont\itshape}
\theorembodyfont{\normalfont}
\newtheorem{lemma}{Lemma}

\newaliascnt{proposition}{lemma}

\aliascntresetthe{proposition}

\newaliascnt{thm}{lemma}
\newtheorem{theorem}[thm]{Theorem}
\aliascntresetthe{thm}

\newaliascnt{corollary}{lemma}
\newtheorem{corollary}[corollary]{Corollary}
\aliascntresetthe{corollary}

\newaliascnt{definition}{lemma}

\aliascntresetthe{definition}

\newaliascnt{claim}{lemma}

\aliascntresetthe{claim}

\theorembodyfont{\normalfont}

\newaliascnt{example}{lemma}
\newtheorem{example}[example]{Example}
\aliascntresetthe{example}

\newaliascnt{remark}{lemma}

\aliascntresetthe{remark}

\newaliascnt{question}{lemma}

\aliascntresetthe{question}

\newaliascnt{conjecture}{lemma}

\aliascntresetthe{conjecture}

\theorembodyfont{\normalfont}
\theoremnumbering{Roman}

\allowdisplaybreaks

\let\originalleft\left
\let\originalright\right
\renewcommand{\left}{\mathopen{}\mathclose\bgroup\originalleft}
\renewcommand{\right}{\aftergroup\egroup\originalright}

\newcommand{\I}              {j}
\newcommand{\E}              {\mathrm{e}}
\DeclareMathOperator{\rank}  {\operatorname{rank}}

\DeclareMathOperator{\diag}  {\mathrm{diag}}

\newcommand{\tensor}[1][{}]           {\mathbin{\otimes_{\scriptscriptstyle{#1}}}}

\newcommand{\Fourier}{\mathcal{F}}

\DeclarePairedDelimiter\ideal\langle\rangle
\reDeclarePairedDelimiterInnerWrapper\ideal{star}{
\mathopen{#1\vphantom{\MTkillspecial{#2}}\kern-\nulldelimiterspace\right.}
#2
\mathclose{\left.\kern-\nulldelimiterspace\vphantom{\MTkillspecial{#2}}#3}}

\newcommand{\DFT}{\operatorname{DFT}}

\newcommand{\TV}{\operatorname{TV}}

\DeclarePairedDelimiter{\norm}{\lVert}{\rVert}

\renewcommand{\vec}{\operatorname{vec}}

\newcommand{\mypar}[1]{{\bf #1.}}

\makeatletter
\let\@@mod\mod
\DeclareRobustCommand{\mod}{\@ifstar\@mods\@@mod}
\def\@mods{\mkern4mu{\operator@font mod}\mkern6mu}
\makeatother

\newcommand{\C}[0]{\mathbb{C}}

\theoremheaderfont{\scshape}
\theorembodyfont{\normalfont}
\theoremstyle{nonumberplain}
\theoremseparator{:}
\theoremsymbol{\hbox{$\boxempty$}}
\newtheorem{proof}{Proof}

\title{Digraph Signal Processing with \\Generalized Boundary Conditions}

\author{Bastian~Seifert,~\IEEEmembership{Member,~IEEE}
  and~Markus~Püschel,~\IEEEmembership{Fellow,~IEEE}
  \thanks{The authors are with the Department of Computer Science, ETH
    Zurich, Switzerland (email: bastian.seifert@inf.ethz.ch,
    pueschel@inf.ethz.ch)}
  \thanks{Manuscript received ???; revised ???}}
    
\begin{document}
%
\maketitle
\begin{abstract}
Signal processing on directed graphs (digraphs) is problematic, since the graph shift, and thus associated filters, are in general not diagonalizable. Furthermore, the Fourier transform in this case is now obtained from the Jordan decomposition, which may not be computable at all for larger graphs. We propose a novel and general solution for this problem based on matrix perturbation theory: We design an algorithm that adds a small number of edges to a given digraph to destroy nontrivial Jordan blocks. The obtained digraph is then diagonalizable and yields, as we show, an approximate eigenbasis and Fourier transform for the original digraph. We explain why and how this construction can be viewed as generalized form of boundary conditions, a common practice in signal processing. Our experiments with random and real world graphs show that we can scale to graphs with a few thousands nodes, and obtain Fourier transforms that are close to orthogonal while still diagonalizing an intuitive notion of convolution. Our method works with adjacency and Laplacian shift and can be used as preprocessing step to enable further processing as we show with a prototypical Wiener filter application.
\end{abstract}

\section{Introduction}
\label{sec:introduction}%

Signal processing on graphs (GSP) extends traditional signal
processing (SP) techniques to data indexed by vertices of graphs and
has found many real world applications, including in analyzing
sensor networks~\cite{Jablonski:2017a}, the detection of neurological
diseases~\cite{Padole.Joshi.Gandhi:2018a}, gene regulatory network
inference~\cite{Pirayre.Couprie.Bidard.Duval.Pesquet:2015a}, 3D point
cloud processing~\cite{Thanou.Chou.Frossard:2016a}, and rating
prediction in video recommendation
systems~\cite{Huang.Marques.Ribeiro:2017a}. See also
\cite{Ortega.Frossard.Kovacevic.Moura.Vandergheynst:2018a} for a
recent overview.

For undirected graphs there are two major variants of GSP that differ
in the chosen shift (or variation) operator: one is based on the
Laplacian~\cite{Shuman.Narang.Frossard.Ortega.Vandergheynst:2013a},
the other is based on the adjacency
matrix~\cite{Sandryhaila.Moura:2013a}. Both are symmetric and thus
diagonalizable with an associated orthogonal Fourier transform. Since
the definition of the shift is sufficient to derive a complete, basic
SP toolset \cite{Pueschel.Moura:2008a} one
obtains in both cases meaningful (but different) notions of spectrum,
frequency response, low and high frequencies, Parseval identities, and
other concepts.

However, in many applications the graph signals are associated with
directed graphs (digraphs). Examples include argumentation framework
analysis~\cite{Dunne.Butterworth:2016a}, predatory-prey
patterns~\cite{Yorke.Anderson:1973a}, big data
functions~\cite{Chui.Mhaskar.Zhuang:2018a}, social
networks~\cite{Kwak.Lee.Park.Moon:2010a}, and epidemiological
models~\cite{Kephart.White:1992a}. In these cases the GSP frameworks
become problematic since non-symmetric matrices may not be diagonalizable. A natural replacement is to use the Jordan normal form (JNF) for the spectral decomposition of the graph~\cite{Sandryhaila.Moura:2014b}.
But the JNF is known to be numerically highly
unstable~\cite{Beelen.VanDooren:1990a} and thus not easy to compute or, for larger graphs, not computable at all. Further, spectral components have now
dimensions larger than one, since no eigenbasis is available, which
complicates the application of SP methods. In the
theory of graph neural networks the non-diagonalizability of digraphs
is problematic as well~\cite{Bronstein.Bruna.LeCun.Szlam.Vandergheynst:2017a}. 

In this paper we propose a novel, practical solution to this problem.
The basic idea is to generalize, in a sense, the well-known concept of
boundary conditions to arbitrary digraphs to make them diagonalizable.
It is best explained using finite discrete-time SP as motivating
example.

\mypar{Motivating example} Imagine we are trying to build an SP
framework for discrete finite-duration time signals using GSP. The
most natural solution is the graph shown in
Fig.~\ref{subfig:PathModel}: it captures the operation of the time
shift and includes no assumptions on the behavior of the signal to the
left and to the right of its support. The associated shift matrix is
shown in Fig.~\ref{subfig:shift0}: it is a single Jordan block and
thus, in a sense, a worst case: it has only the eigenvalue 0, a
one-dimensional eigenspace, and cannot be diagonalized, not even into
block-diagonal form.

\begin{figure}\centering
\hfill
\subfloat[\label{subfig:PathModel}]{
  \vphantom{
	$
	\begin{bsmallmatrix}
	0 & 1 & 0 & \cdots & 0 \\
	0 & 0 & 1 & \cdots & 0 \\
	\vdots & & &  \ddots & \vdots \\
	0 & 0 & 0 & \cdots & 1 \\
	0 & 0 & 0 & \cdots & 0 \\
	\end{bsmallmatrix}$}
  \includegraphics[width=0.45\linewidth]{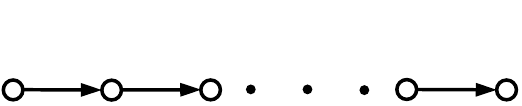}
}
\hfill
\subfloat[\label{subfig:shift0}]{
	$
	\begin{bsmallmatrix}
	0 & 1 & 0 & \cdots & 0 \\
	0 & 0 & 1 & \cdots & 0 \\
	\vdots & & &  \ddots & \vdots \\
	0 & 0 & 0 & \cdots & 1 \\
	0 & 0 & 0 & \cdots & 0 \\
	\end{bsmallmatrix}
	$
}
\hfill
\ 

\hfill
\subfloat[\label{subfig:DFTModel}]{
  \vphantom{$
    \begin{bsmallmatrix}
        0 & 1 & 0 & \cdots & 0 \\
        0 & 0 & 1 & \cdots & 0 \\
        \vdots & & &  \ddots & \vdots \\
        0 & 0 & 0 & \cdots & 1 \\
        1 & 0 & 0 & \cdots & 0 \\
    \end{bsmallmatrix}
    $}
  \includegraphics[width=0.45\linewidth]{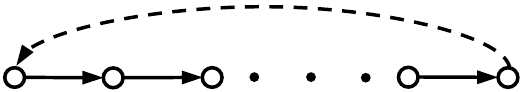}
}
\hfill
\subfloat[\label{subfig:shift1}]{
$
\begin{bsmallmatrix}
0 & 1 & 0 & \cdots & 0 \\
0 & 0 & 1 & \cdots & 0 \\
\vdots & & &  \ddots & \vdots \\
0 & 0 & 0 & \cdots & 1 \\
1 & 0 & 0 & \cdots & 0 \\
\end{bsmallmatrix}
$
}
\hfill
\ 
\caption{The most natural graph model for a finite-duration discrete
  time signal (a) yields the non-decomposable shift matrix in (b).
  Adding one edge (c) yields the well-known circular shift (d) and
  DFT-based spectral analysis. This edge makes
  the support a circle, which is equivalent to assuming the signal as
  periodic.}
\label{fig:DFTModel}
\end{figure}

Indeed, Fig.~\ref{subfig:PathModel} is not the model commonly adopted
but instead Fig.~\ref{subfig:DFTModel}, which adds one edge usually interpreted as a circular boundary condition. Adding this edge (in GSP) makes the graph a circle, and hence signals on this graph are equivalent to periodic signals in DSP.
The associated shift matrix describes the well-known circular shift
(Fig.~\ref{subfig:shift1}). It has an eigendecomposition with distinct
eigenvalues, done by the discrete Fourier transform (DFT). Note that
in almost all applications that use the DFT, {\em the signal is not really periodically extended outside its support}. So, in a
sense, adding the extra edge, or assuming periodicity, can be viewed
as an assumption used to obtain a workable basic SP toolset.

There are a few other aspects worth noting. First, the added edge is
the unique single edge that makes the matrix in Fig.~\ref{subfig:shift0}
diagonalizable and invertible. Further, the eigenvectors of the cyclic shift matrix are approximate eigenvectors of the matrix in Fig.~\ref{subfig:shift0}, and the DFT
diagonalizes this matrix approximately (namely up to a rank-one matrix). We will study these and other aspects in our contribution. 

\mypar{Contributions} The overall contribution of this paper is a
novel approach to make GSP practical on non-diagonalizable digraphs.
For a given digraph, our high-level idea is to add a small number of
edges to make the graph diagonalizable (and also invertible and with
distinct eigenvalues if desired) to obtain a practical form of
spectral analysis. To achieve this we leverage results from
perturbation theory~\cite{Moro.Dopico:2003a,Savchenko:2004a} on the
destruction of Jordan blocks by adding low-rank matrices.

First, we instantiate the perturbation theory to the GSP setting and use
it to design an algorithm that iteratively destroys Jordan blocks by
adding edges. We investigate the consequences for spectral analysis
and show that the added edges can be considered as generalized
boundary conditions in the sense that they add periodic boundary
conditions on subgraphs and increase the number of cycles in the
graph. Second, we provide an efficient implementation of our algorithm
that employs additional techniques to make it numerically feasible and scalable to large
graphs. In particular, the graph Fourier basis we obtain is numerically stable by construction. 

We apply our algorithm to various synthetic and real world graphs showing that usually few edges suffice and that we can process even difficult graphs with a few thousand nodes or close to being acyclic. Finally, we include an application example of a graph Wiener filter enabled by our approach.

\mypar{Related work} The non-diagonalizability in digraph SP is an
important open
problem~\cite[Sec.~III.A]{Ortega.Frossard.Kovacevic.Moura.Vandergheynst:2018a}
and a number of solutions have been proposed. Most approaches aim for
a notion of Fourier basis that circumvents JNF computation at the
price of other GSP properties that are lost.

One idea is to define a different notion of Fourier basis that is orthonormal by construction.
Motivated by the Lov\'asz extension of the graph cut size, \cite{Sardellitti.Barbarossa.DiLorenzo:2017a} defines 
a notion of directed total variation and constructs an orthonormal Fourier basis that minimizes the sum of these. Extending these ideas,
\cite{Shafipour.Khodabakhsh.Mateos.Nikolova:2018a,Shafipour.Khodabakhsh.Mateos.Nikolova:2019a} defines a digraph Fourier basis as the solution of an optimization problem on the Stiefel manifold, minimizing a dispersion function to evenly spread frequencies in the frequency range. Both approaches only work for real signals and yield real Fourier transforms, though a slight modification of the approach was used in~\cite{Shafipour.Khodabakhsh.Mateos.Nikolova:2019a} to make the connection to the circle graph in standard discrete time signal processing. In both cases there is no intuitive notion of convolution in the graph domain anymore, i.e., all filtering now requires the Fourier transform.

Another idea is an approximation of the Fourier basis that almost
diagonalizes the adjacency matrix by allowing small, bounded
off-diagonal entries as proposed in~\cite{Domingos.Moura:2020a}. The
approach is based on the Schur decomposition and the authors solve a
non-convex optimization problem to obtain a numerically stable basis
that can be inverted to compute the Fourier transform. The approach
could be problematic for acyclic digraphs which have zero as the only eigenvalue.

The work in~\cite{Deri.Moura:2017a,Deri.Moura:2017b,Deri.Moura:2017c} maintains the idea of Jordan decomposition but alters the definition of the graph Fourier transform to decompose into Jordan subspaces only, instead of a full JNF. This way a coordinate-free definition of Fourier transform is obtained, which
fulfills a generalized Parseval identity. A method for the inexact, but numerical stable, computation of this graph Fourier transform was proposed in~\cite{Deri.Moura:2017b}. 

Another approach is to change the graph shift operator and thus change the underlying definitions of spectrum and Fourier transform. In~\cite{Furutani.Shibahara.Akiyama.Hato.Aida:2019a} the Hermitian Laplacian matrix is proposed, which is always diagonalizable. The known directed Laplacian was used in~\cite{Singh.Chakraborty.Manoj:2016a} and a scaled version of it, with a detailed study, in \cite{Bauer:2012a}. Both shifts are not diagonalizable, and our proposed method is applicable in both cases.

The work in~\cite{Misiakos.Wendler.Pueschel:2020a} stays within the
framework of \cite{Sandryhaila.Moura:2013a} but identifies the subset
of filters that are diagonalizable. Since these form a subalgebra,
they are generated by one element which can be used as diagonalizable
shift at the price of a smaller filter space. The approach fails for
digraphs with all eigenvalues being zero, i.e., directed acyclic
graphs.

Certain very regular digraphs possess orthonormal Fourier transforms,
e.g., those associated with a directed hexagonal grid~\cite{Mersereau:1979a}, a directed quincunx
grid~\cite{Pueschel.Roetteler:2005a}, or weighted path graphs~\cite{Sandryhaila:12}

Our approach computes an approximate Fourier basis and transform as
some prior work, but is fundamentally different in that it does so by adding a small number of edges to achieve both: stay within the traditional GSP setting and maintain an intuitive notion of convolution.


\section{Graph Signal Processing}
\label{sec:DigraphSP}

In this section we recall the theory of signal processing on graphs, and, in particular, directed graphs (digraphs). We focus on digraphs without edge weights as these are most prone to non-diagonalizable adjacency matrices. However, our approach is applicable to weighted digraphs and discussed later.

\mypar{Directed graphs} A digraph $G = (\mathcal{V},\mathcal{E})$
consists of a set of $n$ vertices $\mathcal{V}$ and a set of edges
$\mathcal{E}\subseteq \mathcal{V}\times\mathcal{V}$. Assuming a chosen
ordering of the vertices, $\mathcal{V} = (v_1,\dots,v_n)$, a digraph
can be represented by its $n\times n$ adjacency matrix $A$ with
entries
\begin{equation}\label{eq:AdjacencyMatrix}
    A_{i,j} =
    \begin{cases}
        1 & \text{if } (v_i,v_j) \in \mathcal{E}, \\
        0 & \text{else.}
    \end{cases}
\end{equation}
We consider graphs with loops, i.e., edges of the form $(v_i,v_i)$ are
allowed. If $A$ is symmetric, i.e., $(v_i,v_j)\in\mathcal{E}$ implies
$(v_j,v_i)\in\mathcal{E}$, the digraph can be viewed as an undirected graph, and hence $A$ is diagonalizable. For other digraphs this may not be the case.

\mypar{Graph signal} A graph signal $s$ on $G$ associates values with
the vertices, i.e., it is a mapping of the form
\begin{equation}\label{eq:GraphSignal}
s: \mathcal{V}\rightarrow\C;\ v_i \mapsto s_i.
\end{equation}
Using the chosen vertex ordering, the graph signal is represented by
the vector $s = (s_i)_{1\leq i\leq n} \in \mathbb{C}^n$.

\mypar{Fourier transform based on adjacency matrix} In GSP based on
\cite{Sandryhaila.Moura:2014b}, the Jordan decomposition of $A$,
\begin{equation}
    \label{eq:JordanDecompositionAdjacency}
    A = V \cdot J \cdot V^{-1},
\end{equation}
where $J$ is in Jordan normal form (JNF), yields $\Fourier = V^{-1}$
as the graph Fourier transform of the graph. The graph Fourier
transform of a graph signal $s$ is
\begin{equation}
    \label{eq:GraphFourierTransform}
    \hat{s} = \Fourier s.
\end{equation}
The frequencies are ordered by total variation, defined as
\begin{equation}
\label{eq:GraphTotalVariation}
\TV_A(v) = \norm{v - \tfrac{A}{|\lambda_{\max}|} v}_1,
\end{equation}
where $\lambda_{\max}$ is the eigenvalue of $A$ with largest
magnitude.

Note that the computation of the JNF is numerically
unstable~\cite{Beelen.VanDooren:1990a}. For example, $
\begin{bsmallmatrix}
    0 & 1 \\
    0 & 0
\end{bsmallmatrix}
$ 
is in JNF, whereas 
$
\begin{bsmallmatrix}
    \epsilon & 1 \\
    0 & -\epsilon
\end{bsmallmatrix}
$ 
is diagonalizable for every $\epsilon\neq 0$. Thus symbolic computation is needed, which, however, becomes too expensive for graphs with hundreds or more nodes.

\mypar{Fourier transform based on Laplacian} An alternative approach to GSP is based on the Laplacian of a graph. For digraphs, several variants of Laplacians have been proposed including the most common directed
Laplacian~\cite{Hein.Audibert.vonLuxburg:2007a}, the normalized
Laplacian~\cite{Bauer:2012a}, the random-walk
Laplacian~\cite{Chung:2005a}, or the magnetic
Laplacian~\cite{Shubin:1994a}. The last two variants are always
diagonalizable, as they are either symmetric or Hermitian.

In~\cite{Shuman.Narang.Frossard.Ortega.Vandergheynst:2013a} the graph
Fourier transform for undirected graphs was defined using the
eigendecomposition of the Laplacian. For the extension of this
framework to directed graphs, \cite{Singh.Chakraborty.Manoj:2016a}
thus uses the directed Laplacian
\begin{equation}\label{eq:DirectedLaplacian}
    L = D - A,
\end{equation}
where $D$ is the diagonal matrix of either in- or out-degrees. The
Jordan decomposition of the directed Laplacian
\begin{equation}
    \label{eq:JordanDecompositionLaplacian}
    L = V \cdot J \cdot V^{-1},
\end{equation}
is then used to define the graph Fourier transform $\Fourier = V^{-1}$
as before. The frequencies are ordered in~\cite{Singh.Chakraborty.Manoj:2016a} by graph total variation as well.

Our focus will be GSP based on \eqref{eq:JordanDecompositionAdjacency} but we will also instantiate our approach to GSP based on \eqref{eq:JordanDecompositionLaplacian} to which it is equally applicable.



\section{Generalized Boundary Conditions for Digraphs}
\label{sec:FixingSpectra}

In the introduction we gave a motivating example for the contribution
in this paper: a "bottom-up" explanation of the cyclic boundary
condition (or, equivalently, periodicity) inherently assumed with
DFT-based spectral analysis.\footnote{More common is what one could
call the "top-down" explanation for periodicity, which naturally
arises, for example, when sampling the spectrum of continuous signals.} Namely, in GSP terms, the cyclic boundary condition is the minimal addition of edges to the graph in Fig.~\ref{subfig:PathModel} to obtain a proper spectrum
with distinct eigenvalues.

In this section we extend this basic idea and construction to
arbitrary digraphs: Given a digraph, our goal is to add the minimal
number of edges that make the digraph diagonalizable. In matrix terms
this means adding to the adjacency matrix $A$ a low-rank adjacency matrix $B$
containing the additional edges, such that $A+B$ is diagonalizable. The same technique can be used to make $A$ also invertible or the eigenvalues distinct. An analogous construction can be done for the directed Laplacian by ensuring that the Laplacian structure is preserved.

Our approach builds on results from matrix perturbation theory on the
destruction of Jordan blocks under low-rank changes of a matrix.

We first introduce the needed results from perturbation theory and then instantiate them in the GSP setting to design an algorithm that destroys Jordan blocks by adding edges to graphs. We provide a number of theoretical results and explain in which way one may consider the added edges as generalized boundary conditions. Accompanying the theory we provide small, illustrating examples.

\subsection{Results from Perturbation Theory}\label{subsec:PerturbationTheory}

We recall some terminology. For a matrix $M$, $v$ is a right
eigenvector if $Mv = \lambda v$ and $u$ a left eigenvector if
$u^TM = \lambda u^T$, i.e., $u$ is an eigenvector of the transpose
$M^T$. $M$ and $M^T$ have the same JNF.

Let $J = V^{-1}MV$ be in JNF. Then $J$ is a block-diagonal matrix consisting of Jordan blocks of the form
\begin{equation}\label{eq:jb}
\begin{bsmallmatrix}
\lambda & 1 & 0 & \cdots & 0 \\
0 & \lambda & 1 & \cdots & 0 \\
\vdots & & &  \ddots & \vdots \\
0 & 0 & 0 & \cdots & 1 \\
0 & 0 & 0 & \cdots & \lambda \\
\end{bsmallmatrix},
\end{equation}
where $\lambda$ is an eigenvalue. Each eigenvalue can have multiple such
blocks and of different size. The Jordan basis (columns of $V$) associated with each
block includes exactly one right eigenvector, which is in first
position of the block, and exactly one left eigenvector (a row in $V^{-1}$), which is in
last position.

\mypar{Matrix perturbation and Jordan blocks} Our work builds on results from perturbation theory \cite{Moro.Dopico:2003a,Savchenko:2004a} that study the effect on the Jordan blocks when perturbing a given matrix
$M \in \mathbb{C}^{n \times n}$ by adding a low rank matrix $B$. We will use the following main result, which can also be found in~\cite{Hoermander.Melin:1994a} without the explicit condition on matrices. We work with the exposition in~\cite{Moro.Dopico:2003a} in a
slightly adapted formulation.

\begin{theorem}[\cite{Moro.Dopico:2003a,Savchenko:2004a}]\label{thm:DestroyingJordanBlocks} 
    Assume the different sizes of the Jordan blocks to a given
    eigenvalue $\lambda$ of $M$ are $f_1>f_2>\dots>f_t$ and that the
    Jordan blocks are ordered accordingly. Let $r_s$ be the number of
    blocks of size $\geq f_s$, $s = 1,\dots,t$, and set $r_0 = 0$. For
    each $r_s$ denote the associated left and right eigenvectors (one
    per block) with $u_k^T$ and $v_k$, $k = 1,\dots,r_s$, s.t. $u_i^T
    v_j = \delta_{ij}$. Let $B$ be a
    matrix of rank $\rho$, with $r_{s-1} < \rho \leq r_{s}$ and define
    \begin{equation}
        \label{eq:GenericityConditionAdds}
        \Phi_s = 
        \begin{bmatrix}
            u_1^T \\
            \vdots \\
            u_{r_s}^T
        \end{bmatrix}
        B
        \begin{bmatrix}
            v_1 & \ldots & v_{r_s}
        \end{bmatrix} \in \mathbb{C}^{r_s \times r_s}.     
    \end{equation}
    We denote with $\Phi_{s-1}$ the upper-left block of dimension $r_{s-1}$ of
    $\Phi_s$. $\Phi_0$ is considered as the empty submatrix. If
    \begin{equation}\label{eq:GenericityCondition}
        \sum_{\substack{\phi\\ \phi\text{ principal $\rho\times\rho$ submatrix}\\\text{of $\Phi_s$ containing $\Phi_{s-1}$}}} \det(\phi) \not= 0,   
    \end{equation}
then the Jordan blocks of $M+B$ for $\lambda$ are those of $M$ minus the $\rho$ largest
    ones. A principle submatrix is obtained by deleting rows and columns with the same indices.
\end{theorem}

If real or complex matrices are concerned then a random matrix $B$
will satisfy~\eqref{eq:GenericityCondition} for all eigenvalues with
probability~1. This so-called generic case was the purpose of the study in
\cite{Moro.Dopico:2003a}. In our case, neither the matrices $M$ nor
the desired $B$ (to add edges) are generic since they have only entries 0 or 1 and thus constitute finite sets.

Note that the condition in the theorem is sufficient but not necessary. Further,  \eqref{eq:GenericityCondition} is a statement about the Jordan blocks to \emph{one} eigenvalue and does not state what happens to Jordan blocks of other eigenvalues, which can be destroyed as well (generic case), or remain untouched, or even be enlarged.

In general, destroying a Jordan block for $\lambda$ yields new eigenvalues and the basis for the JNF also changes when adding $B$ to $M$.

Theorem~\ref{thm:DestroyingJordanBlocks} allows the destruction of all Jordan blocks to one eigenvalue with a properly chosen $B$, but the condition is complex. Thus later, we prefer to do so iteratively, one block at a time, with matrices $B$ of rank $\rho = 1$. This means all $\phi$ in \eqref{eq:GenericityCondition} have size $1\times 1$, which avoids determinant computations for simplicity and better numerical stability. The following corollary considers this special case. Note that $\rho = 1$ implies that $s = 1$ such that $0 = r_0 < \rho \leq r_1$.
	
\begin{corollary}\label{corollary:DestroyConditionEasy}%
  For a given eigenvalue $\lambda$ of $M$ let $u_{1},\dots,u_{r_1}$ be the left and $v_{1},\dots,v_{r_1}$ be the right eigenvectors of the Jordan blocks of the largest size $f_1$. Let $B$
  be a matrix of rank $1$. If
  \begin{equation}
      \label{eq:GenericityConditionEasy}
      \sum_{k=1}^{r_1} u_{k}^T B v_{k} \not= 0,
  \end{equation}
  then adding $B$ to $M$ destroys one of the largest Jordan blocks
  to $\lambda$ of $M$.
\end{corollary}
\begin{proof}
  In this special case of Theorem~\ref{thm:DestroyingJordanBlocks}
  we have $\rho = 1 \leq r_1$. The principle $1\times 1$ submatrices $\phi$ of
  \begin{equation*}
      \Phi_1 = 
    \begin{bmatrix}
        u_1^T \\
        \vdots \\
        u_{r_1}^T
    \end{bmatrix}
    B
    \begin{bmatrix}
        v_1 & \ldots & v_{r_1}
    \end{bmatrix} 
  \end{equation*}
containing the empty matrix $\Phi_0$ correspond exactly to the diagonal elements of $\Phi_1$, which yields the result.
\end{proof}

Fig.~\ref{subfig:shift0} is a very simple example since it is already
in JNF with only one Jordan block. The right eigenvector for the block is $u_1 = [1,0,\dots,0]^T$
and the left eigenvector is $v_1 = [0,0,\dots,1]$. The matrix $B$
containing the added edge in position $(n,1)$ indeed satisfies
\eqref{eq:GenericityCondition}:
\begin{equation}
    \label{eq:TimeModelGenericityCondition}
    [0,0,\dots,1] \cdot B \cdot [1,0,\dots,0]^T = 1,
\end{equation}
and is the only matrix $B$ adding one edge with this property.

\mypar{Behavior of new eigenvalues} The following result shows how the eigenvalues change under a rank-one perturbation. It can be easily proved using the matrix determinant lemma but is not practical for large scale graphs.
\begin{lemma}
\label{lemma:NewEigenvaluesAfterEdgeAdding}%
Let $B = ab^T$ be a rank-one matrix. Then the new eigenvalues of the
perturbed matrix $M + B$ are the solutions to the equation
\begin{equation}
    \label{eq:NewEigenvaluesAfterEdgeAdding}
    b^{T} (x I - M)^{-1} a = 1.
\end{equation}
The left-hand side is a rational function, hence the eigenvalues are
given by the roots of a polynomial.
\end{lemma}
For the example in Fig.~\ref{fig:DFTModel}, \eqref{eq:NewEigenvaluesAfterEdgeAdding} becomes $1/x^n = 1$, i.e., the new eigenvalues are exactly the $n$th roots of unity, as expected.

The literature also provides bounds on the distance between old and new (under low-rank perturbation) eigenvalues (e.g.,~\cite[Thm.~8]{Kahan.Parlett.Jiang:1982a}), but we found them to be loose and not of practical value in our application scenario.

Finally, \cite[Thm.~6.2]{Ran.Wojtylak:2012a} shows that for real or
complex matrices, in the generic case, $M + B$ has no repeated
eigenvalues, which are not already eigenvalues of $M$.

\subsection{Adding Edges to Destroy Jordan blocks}\label{subsec:AddingEdgesToKillJordanBlocks}%

Our goal is to perturb a directed graph by adding edges to destroy the
Jordan blocks of its adjacency matrix and Corollary~\ref{corollary:DestroyConditionEasy} will be our main tool. First, we establish the viability of this approach, meaning it is always possible to find a matrix $B$ adding one edge that satisfies \eqref{eq:GenericityConditionEasy}.

In the following, we use the column-wise vectorization of a matrix
$B \in \mathbb{C}^{m \times n}$:
$\vec(B) = (b_{1,1}, \dots, b_{m,1}, b_{1,2}, \dots, b_{m,n})^T$.
Vectorization satisfies $\vec(ABC) = (C^T \tensor A) \vec(B)$ for
matrices of compatible dimensions, where $\tensor$ is the Kronecker
product.

\begin{theorem}\label{prop:AddingOneEdgeToDestroyOneBlock}%
Adding or deleting one edge is sufficient to destroy the largest
Jordan block of an adjacency matrix for a chosen eigenvalue $\lambda$.
\end{theorem}
\begin{proof} 
Let $u_1,\ldots,u_r$ and
      $v_1, \ldots, v_r$ be the left and right eigenvectors
      of the largest Jordan blocks for the eigenvalue $\lambda$,
      respectively. Then \eqref{eq:GenericityConditionEasy} can be written as
      \begin{equation*}
        \begin{split}
        0 &\neq \sum_{k=1}^r u_k^T B v_k \\
        &= \sum_{k=1}^r \vec(u_k^T B v_k)
        = \left(\sum_{k=1}^r v_k^T \tensor u_k^T\right) \vec(B).             
        \end{split}          
    \end{equation*}
    Since the $u_k$ and the $v_k$ are linear independent, the same
    holds for the set of the $v_k \tensor u_k$. Thus,
    $w^T = \sum_{k=1}^r v_k^T \tensor u_k^T$ is a nonzero row vector and
    for an adjacency matrix it is enough (and always possible) to set
    exactly one entry (which depends on the $u_k$ and $v_k$) of $B$ to
    $1$ to ensure that the result is nonzero. The number of nonzero
    elements in $w$ is the number of choices. If for each choice,
    $A$ already contains the edge, we can instead delete an edge,
    choosing $-1$ as entry in~$B$. 
\end{proof}
It is not possible to strengthen the hypothesis to destroying Jordan
blocks by only adding edges in each case. A counter example is the
complete graph, which, however, has only Jordan blocks of size 1. 
In our experiments with (the most relevant) sparse graphs, we never encountered the case that a Jordan block could not be destroyed by adding an edge.

\mypar{Basic algorithm} Using
Theorem~\ref{prop:AddingOneEdgeToDestroyOneBlock} we can formulate the
basic mathematical algorithm to make a digraph adjacency matrix $A$
diagonalizable by adding edges
(Fig.~\ref{algo:AbstractFixingAlgorithm}). The algorithm is iterative,
adding one edge in each step as described in
Theorem~\ref{prop:AddingOneEdgeToDestroyOneBlock},
$A \rightarrow A + B$, to destroy the largest Jordan block. $B$ has
only one entry 1. Note that in the case that all edges that are eligible for adding already exist in the graph, we choose to add a random edge instead of removing an edge. This way the algorithm is guaranteed to terminate as discussed below.

\begin{figure}
    \centering
    \begin{algorithmic}
        \Function{\textbf{DestroyAllJordanBlocks}}{$A$}
        \While {$A$ not diagonalizable}
        \State $u_1, \ldots, u_r \leftarrow$ left EVs
        to largest Jordan blocks 
        \State $v_1, \ldots, v_r \leftarrow$ right
        EVs to largest  Jordan blocks
        \If{$\exists (i,j)$ s.t. $\sum_k u_{k,j} v_{k,i} \not= 0$
          and $A_{i,j} = 0$}
        \State $A_{i,j} \leftarrow 1$
        \Else
        \State select $(i,j)$ random s.t. $A_{i,j} = 0$
        \State $A_{i,j} \leftarrow 1$
        \EndIf
        \EndWhile
        \State \textbf{return} $A$
        \EndFunction
    \end{algorithmic}
    \caption{The mathematical algorithm to obtain a diagonalizable digraph. $u_{k,j}$ is the $j$th element of $u_k$.}
    \label{algo:AbstractFixingAlgorithm}%
\end{figure}
For a practical implementation, various additional details need to be
considered that we discuss later.

\mypar{An example} To illustrate
Alg.~\ref{algo:AbstractFixingAlgorithm} we provide a detailed example.
\begin{example}\label{example:7NodesAndTwoJordanBlocks}
%
\begin{figure}\centering
    \subfloat[\label{subfig:SmallExampleGraph}]{
      \includegraphics[width=0.45\linewidth]{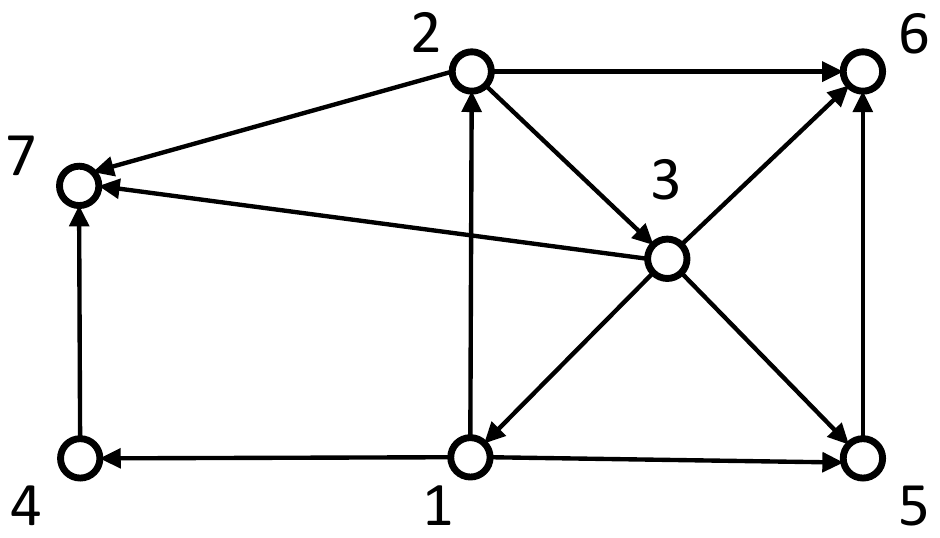}
    }\
    
    \subfloat[\label{subfig:JordanForm7NodesGraph}]{
      $
      \begin{bsmallmatrix}
          0 & 1 &  \\
          0 & 0 &  \\
          & & 0 & 1 &  \\
          & & 0 & 0 &  \\
          & & & & 1 \\
          & & & & & \omega_3^2 \\
          & & & & & & \omega_3
      \end{bsmallmatrix}
      $
    }
    \hspace{0.5em}
    \subfloat[\label{subfig:JordanTransform7NodesGraph}]{
      $ 
      \begin{bsmallmatrix}
          -1 & -1 & 0 & -1 & 1 & \omega_3^2 & \omega_3 \\
          -1 & -1 & -1 & 0 & 1 & \omega_3 & \omega_3^2 \\
          0 & -2 & 0 & -2 & 1 & 1 & 1 \\
          0 & 0 & 1 & 0 & 0 & 0 & 0 \\
          1 & 0 & 0 & 0 & 0 & 0 & 0 \\
          0 & 1 & 0 & 0 & 0 & 0 & 0 \\
          0 & 0 & 0 & 1 & 0 & 0 & 0
      \end{bsmallmatrix}
      $
    }
    \caption{The graph shown in (a) has the Jordan normal form
      $J = V^{-1} A V$ shown in (b). The matrix $V$ of generalized
      eigenvectors is shown in (c), with $\omega_3 = \exp(-2 \pi \I / 3)$.}
    \label{fig:7NodesExample}
\end{figure}

We consider the graph in Fig.~\ref{fig:7NodesExample}, which has the
characteristic polynomial $p(x) = x^4(x^3-1)$ and two Jordan blocks of size 2
for eigenvalue 0. We apply Alg.~\ref{algo:AbstractFixingAlgorithm}.
The right eigenvectors for to the Jordan blocks of size two are
the first and third column of $V$:
\begin{equation}\label{eq:7NodesGraphRightEigenvectorFirstBlock}
    \begin{split}
        v_1 &=
        \begin{bsmallmatrix}
            -1 & -1 & 0 & 0 & 1 & 0 & 0 
        \end{bsmallmatrix}^T, \\
        v_2 &=
        \begin{bsmallmatrix}
            0 & -1 & 0 & 1 & 0 & 0 & 0 
        \end{bsmallmatrix}^T.        
    \end{split}
\end{equation}
The corresponding left eigenvectors are the second and fourth row of
$V^{-1}$:
\begin{equation}\label{eq:7NodesGraphLeftEigenvectorFirstBlock}
    \begin{split}
        u_1^T &=
        \begin{bsmallmatrix}
            0 & 0 & 0 & 0 & 0 & 1 & 0 
        \end{bsmallmatrix}, \\
        u_2^T &=
        \begin{bsmallmatrix}
            0 & 0 & 0 & 0 & 0 & 0 & 1 
        \end{bsmallmatrix}.
    \end{split}
\end{equation}
Thus~\eqref{eq:GenericityCondition} takes the form
\begin{equation}\label{eq:7NodesGraphConditionFirstBlock}
    0\neq u_1^TBv_1 + u_2^T B v_2
    = -b_{6,1} - b_{6,2} + b_{6,5} - b_{7,2} + b_{7,4},
\end{equation}
and we have five choices. We choose $b_{6,1} = 1$, as shown in
Fig.~\ref{subfig:SmallExampleGraphFirstEdge}, which defines $B$. 

Note the effect on the JNF (Fig.~\ref{subfig:JordanForm7NodesGraph}) when $B$ is added: $V^{-1}(A+B)V = J + V^{-1}BV$ with 
\begin{equation}\label{eq:7NodesGraphSpectralFormEdge}
V^{-1} B V =
\begin{bsmallmatrix}
    0 & 0 &  0 & 0 & 0 & 0 & 0 \\
    -1 & -1 & 0 & -1 & 1 & \omega_3^2 & \omega_3 \\
    0 & 0 &  0 & 0 & 0 & 0 & 0 \\
    0 & 0 &  0 & 0 & 0 & 0 & 0 \\
    -\tfrac{4}{3} & -\tfrac{4}{3} & 0 & -\tfrac{4}{3} &
    \tfrac{4}{3} & \tfrac{4}{3} \omega_3^2 & \tfrac{4}{3}
    \omega_3 \\
    -\tfrac{1}{3} & -\tfrac{1}{3} & 0 & -\tfrac{1}{3} &
    \tfrac{1}{3} & \tfrac{1}{3} \omega_3^2 & \tfrac{1}{3}
    \omega_3 \\
    -\tfrac{1}{3} & -\tfrac{1}{3} & 0 & -\tfrac{1}{3} &
    \tfrac{1}{3} & \tfrac{1}{3} \omega_3^2 & \tfrac{1}{3}
    \omega_3 \\
\end{bsmallmatrix}.
\end{equation}
The addition of this matrix to the Jordan form of $A$ modifies all eigenvalues, except for the remaining Jordan block for eigenvalue 0. 

\begin{figure}\centering
    \subfloat[\label{subfig:SmallExampleGraphFirstEdge}]{
      \includegraphics[width=0.45\linewidth]{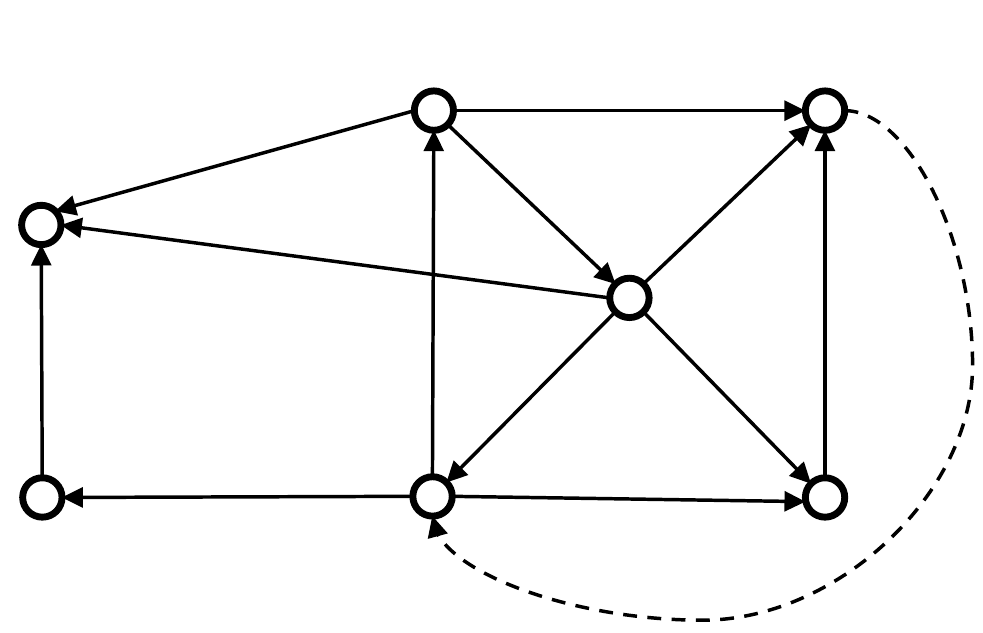}
    }
    \hfill
    \subfloat[\label{subfig:SmallExampleGraphSecondEdge}]{
      \includegraphics[width=0.45\linewidth]{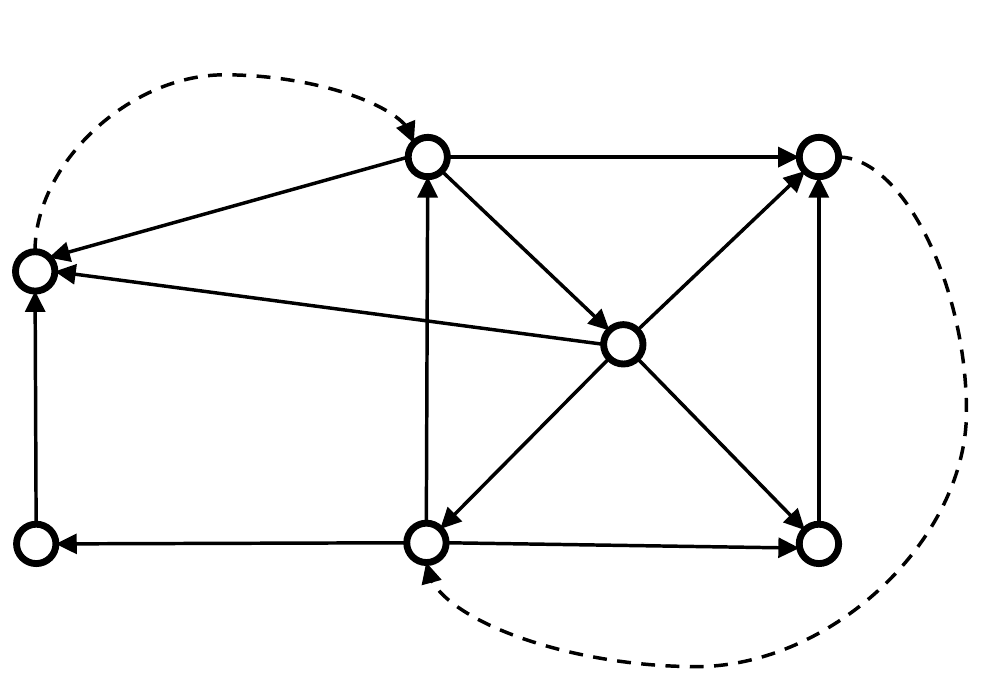}
    }
    \hfill
    \caption{The modified graph with the first (a) and second (b)
      edge added by the proposed abstract algorithm.}
    \label{fig:7NodesExampleSteps}
\end{figure}

For the remaining Jordan block for eigenvalue $0$ in the modified graph, the right and left eigenvectors are given, respectively, by
\begin{equation*}
v_2 =
\begin{bsmallmatrix}
    0 & -1 & 0 & 1 & 0 & 0 & 0 
\end{bsmallmatrix}^T,\quad
u_2^T =
\begin{bsmallmatrix}
    0 & 0 & 0 & 0 & 0 & 0 & 1 
\end{bsmallmatrix}.
\end{equation*}
Condition~\eqref{eq:GenericityCondition} takes the form
\begin{equation}\label{eq:7NodesGraphConditionSecondBlock}
u_2^TBv_2 = - b_{7,2} + b_{7,4}.
\end{equation}
We add the edge $b_{7,2} = 1$ and obtain the graph in
Fig.~\ref{subfig:SmallExampleGraphSecondEdge}. The characteristic
polynomial is now $p(x) = x^7-x^5-4x^4-x^3-2x^2-1$, which yields pairwise
different eigenvalues (Fig.~\ref{fig:7NodesExampleEigenvalues}). 

\begin{figure}
    \centering
    \includegraphics[width=0.8\linewidth]{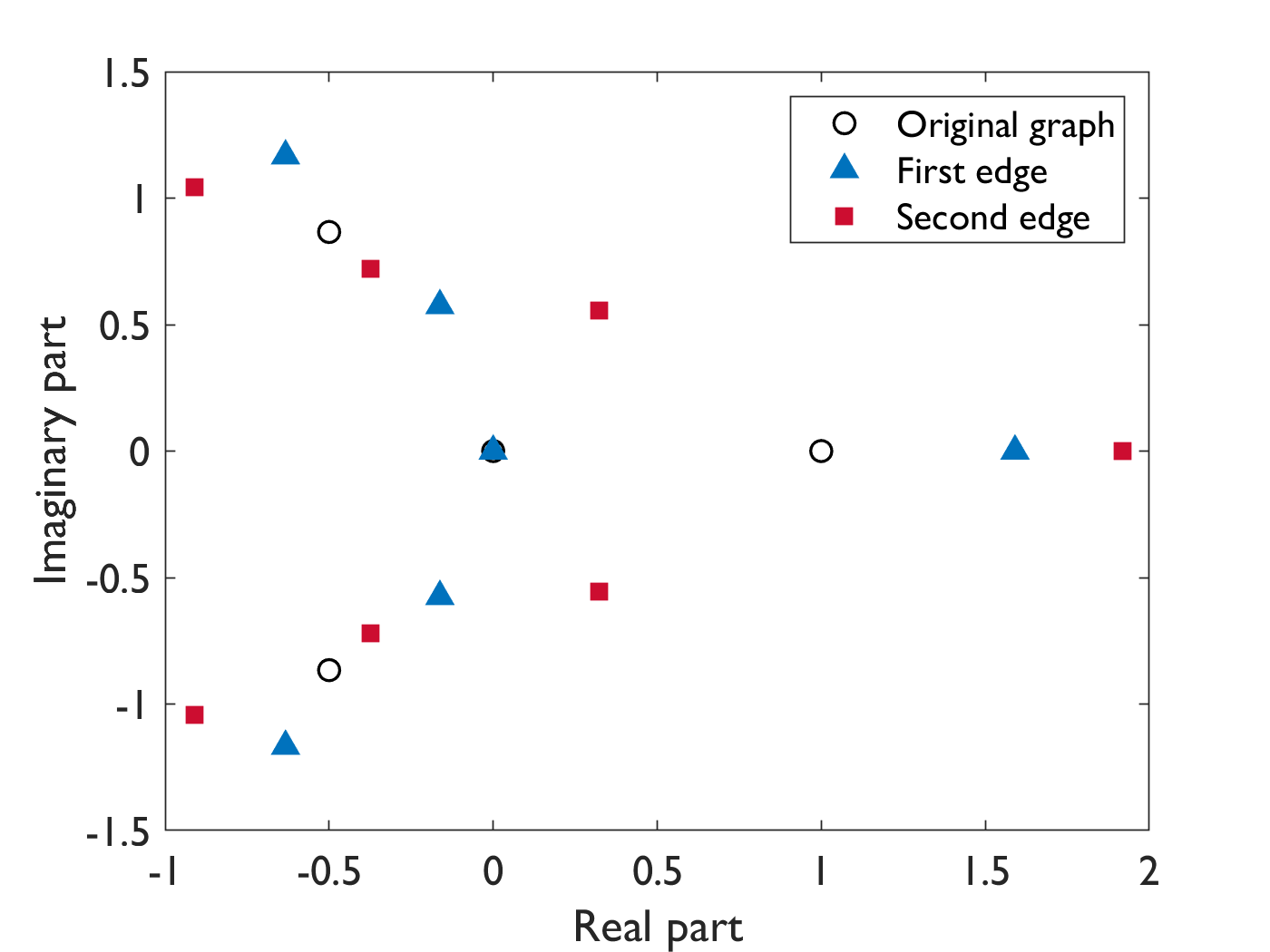}    
    \caption{The eigenvalues of the original graph on 7 nodes (black
      circles), after the first edge was added (blue triangles), and
      after the second edge was added (red squares).}
    \label{fig:7NodesExampleEigenvalues}
\end{figure}
\end{example}

\subsection{Further Properties and Discussion}

We discuss various properties of our basic algorithm, the results it produces, and further extensions. In particular, we provide an explanation for terming the added edges generalized boundary conditions.

\mypar{Minimal number of edges}
Theorems~\ref{prop:AddingOneEdgeToDestroyOneBlock} and
\ref{thm:DestroyingJordanBlocks} give an immediate lower bound for the
number of edges to destroy all Jordan blocks: it is the maximal number
of Jordan blocks of size larger than one over all eigenvalues. The
bound is then achieved if destroying this maximum number of blocks
happens to destroy the Jordan blocks of all other eigenvalues as well.
For real and complex matrices, this would hold in the generic case.
For adjacency matrices, in general, it does not. 

A trivial upper bound is the number of edges needed to make the graph symmetric. This is of course not the purpose of our work, a large number, and not the type of edges found by our algorithm in practice.

\mypar{Termination} The algorithm in Fig.~\ref{algo:AbstractFixingAlgorithm} always terminates since it adds an edge in every step, either one which destroys one Jordan
block or a random one. In the worst case it would reach the unweighted complete graph, which is diagonalizable. Again, we note that in our extensive experiments on sparse graphs we never saw the case of a random edge, i.e., in every step a Jordan block got destroyed. Note that the potential (non-generic) case that a new Jordan block is created if another is destroyed thus also poses no problem for termination.

\mypar{Invertible adjacency matrix} Since the adjacency matrix is
considered as shift in the GSP of~\cite{Sandryhaila.Moura:2013a}, it
may be desirable that it is invertible. Our algorithm can be used for
this purpose by also destroying all Jordan blocks for the
eigenvalue zero, including those of size one.

\mypar{Approximate eigenvectors and Fourier transform} Our algorithm takes as input an adjacency matrix $A$ and outputs a diagonalizable $A+B$, where $B$ contains all the added edges, say $k$ many. As we show now, the eigenvectors of $A+B$ are, in a sense, approximate eigenvectors of $A$ and the same holds for the Fourier transform of $A+B$.
\begin{lemma}\label{thm:NewEigenvectorsForOldGraph}%
If $v$ is an eigenvector of $A + B$ to the eigenvalue $\lambda$ then
\begin{equation}\label{eq:NewEigenvectorsForOldGraph}
||Av - \lambda v||_0 \leq k,
\end{equation}
where $||\cdot||_0$ is the $\ell_0$-pseudonorm that counts the entries $\neq 0$.
\end{lemma}
\begin{proof}
Let $I$ be the index set of zero rows of $B$, $|I|\geq n-k$. Then
\begin{equation*}
(Av)_{i \in I}
= ( (A + B)v)_{i \in I}
= (\lambda v)_{i \in I}
= \lambda (v)_{i \in I},        
\end{equation*}
as $B$ has no effect on the entries corresponding to $I$.
\end{proof}
As a consequence, $A$ also gets diagonalized approximately by the Fourier transform $\Fourier$ of $A+B$ in the following sense.
\begin{lemma}If $\Fourier(A+B)\Fourier^{-1} = D$ (diagonal), then
$$
\Fourier A\Fourier^{-1} = D - \Fourier B\Fourier^{-1},
$$
is diagonal up to a matrix of rank $k$.
\end{lemma}
For example, the $\DFT$ diagonalizes the matrix in Fig.~\ref{subfig:shift0} up to a dense rank-one matrix, which is the outer product of the last column of $\DFT$ with the first row of $\DFT^{-1}$.

\mypar{New edges as generalized boundary conditions} We explain why the edges added by our algorithm to destroy Jordan blocks may be considered as generalized boundary conditions. In the example in Fig.~\ref{fig:DFTModel} we saw that the added edge created a cycle. Intriguingly, this observation generalizes: there is an intrinsic relationship between diagonalizability (and invertibility) of $A$ and the occurrence of cycles.

To do so, we first need the following theorem for digraphs that explains the connection between the coefficients of the characteristic polynomial of $A$ and the simple cycles of the graph. We recall that a cycle is simple if all the vertices it contains are different. Further, $H$ is called a subgraph of $G$ if it contains a subset of the vertices and edges of $G$.

\begin{theorem}[\cite{Coates:1959a}]\label{theorem:CoefficientTheoremDigraphs}%
    Let $G$ be a graph and denote by $\mathcal{H}_i$ the set of all
    subgraphs of $G$ with exactly $i$ vertices and consisting of a disjoint union
    of simple directed cycles (equivalently, $\mathcal{H}_i$ consists of all subgraphs with $i$ nodes, each of which has indegree and outdegree $=1$). Then the coefficients of the characteristic
    polynomial of $G$
    \begin{equation}
        \label{eq:CharPolyCoefficientTheorem}
        p_G(x) = x^n + a_{n-1} x^{n-1} + \dots + a_0 
    \end{equation}
    have the form
    \begin{equation}
        \label{eq:CyclesCoefficientTheorem}
        a_i = \sum_{H \in \mathcal{H}_{n-i}} (-1)^{c(H)},
    \end{equation}
    where $c(H)$ is the number of cycles $H$ consists of.
\end{theorem}
For example the graph in Fig.~\ref{subfig:SmallExampleGraphSecondEdge}
has four subgraphs on three vertices consisting of simple cycles shown
in Fig.~\ref{fig:7NodesSubgraphs}. Hence, by
Theorem~\ref{theorem:CoefficientTheoremDigraphs}, its characteristic
polynomial has the term $-4 x^4$, which is indeed the case.

\begin{figure}
    \centering
    \includegraphics[width=\linewidth]{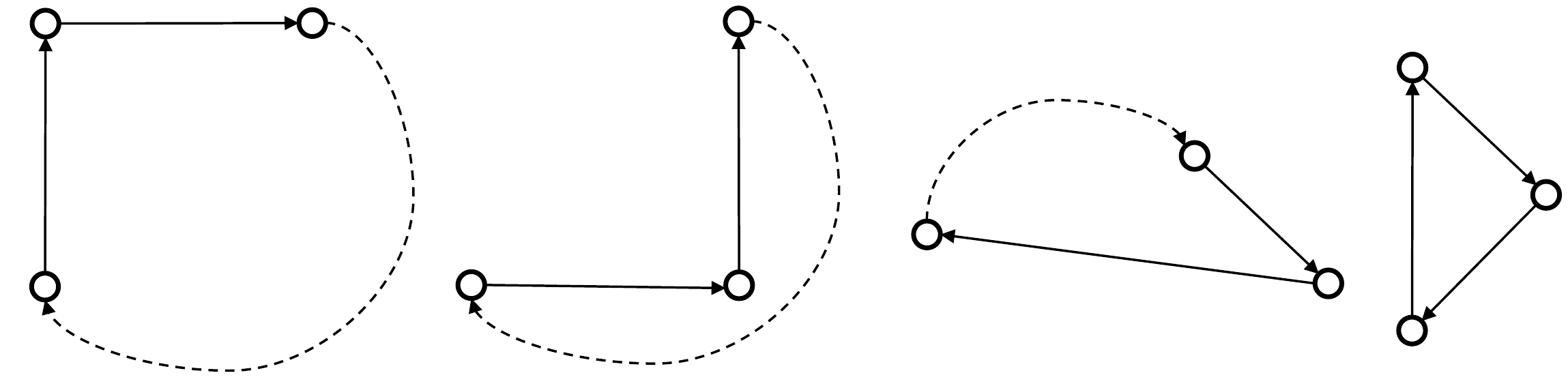}
    \caption{The 4 subgraphs in $\mathcal{H}_3$  of the graph in
      Fig.~\ref{subfig:SmallExampleGraphSecondEdge}. The added edges are dotted.}
    \label{fig:7NodesSubgraphs}
\end{figure}

It is clear that adding edges cannot reduce the number of cycles in a graph. Theorem~\ref{theorem:CoefficientTheoremDigraphs} implies that if an added edge is not part of any cycle, it will not change the characteristic polynomial. But Algorithm~\ref{algo:AbstractFixingAlgorithm} does. Hence we get the following corollary.
\begin{corollary}
\label{theorem:NewEdgesAddNewSimpleCycles}%
Each edge that Algorithm~\ref{algo:AbstractFixingAlgorithm} adds to a
graph to destroy a Jordan block introduces additional simple
cycles.
\end{corollary}
The added edges by Algorithm~\ref{algo:AbstractFixingAlgorithm} thus add periodic boundary conditions to certain subgraphs (see the example in Fig.~\ref{fig:7NodesSubgraphs}). Thus we term them generalized boundary conditions.

One could consider vertices with indegree or outdegree $=0$ (sources or sinks) as boundaries. Such vertices make $A$ non-invertible, i.e., produce eigenvalues $=0$. Our algorithm can be used to remove the eigenvalue 0 by adding edges, thus making $A$ invertible and removing sinks and sources. Note that the added edges in Fig.~\ref{subfig:SmallExampleGraphSecondEdge} achieved exactly that.

\mypar{Directed acyclic graphs} The class of directed acyclic graphs (DAGs) without self-loops constitutes in a sense the worst-case class for signal processing on
graphs. A DAG represents a partial order, and thus the vertices can be topologically sorted to make $A$ triangular, i.e., the characteristic polynomials is $p(x) = x^n$ and the only eigenvalue is $0$. Equivalently, no edge is part of a cycle and thus, by Theorem\ref{theorem:CoefficientTheoremDigraphs}, all edges can be removed without changing $p(x)$, which yields the same result $p(x) = x^n$.

As an example consider the product graph of two directed path graphs
in Fig.~\ref{fig:DirectedGrid}, which is a DAG. The vertices are numbered from 1 to 9 starting in the bottom left.
The JNF consists of three Jordan
blocks of sizes one, three, and five. Applying the proposed algorithm
yields the condition to destroy the largest Jordan block as
\begin{equation}
\label{eq:DirectedGridFirstBlockCondition}
6 b_{9,1} \not= 0,
\end{equation}
while the condition to destroy the second Jordan block is
\begin{equation}
\label{eq:DirectedGridSecondBlockCondition}
\tfrac{1}{2}( b_{6,2} - b_{6,4} - b_{8,2} + b_{8,4}) \not= 0.
\end{equation}
Hence adding the edge $(9,1)$ and any of the ones occurring
in~\eqref{eq:DirectedGridSecondBlockCondition} makes the graph
diagonalizable with distinct eigenvalues. One solution is shown in
Fig.~\ref{fig:DirectedGrid}. Adding one more edge, which can be obtained
from the condition
\begin{equation}
    \label{eq:DirectedGridThirdBlockCondition}
    \tfrac{1}{3}( b_{3,3} - b_{3,5} + b_{3,7} - b_{5,3} + b_{5,5} -
    b_{5,7} + b_{7,3} - b_{7,5} + b_{7,7}) \not= 0,
\end{equation}
can destroy the last block for eigenvalue 0 to make $A$ invertible.

\begin{figure}
	\centering
    \subfloat[b][\label{subfig:DirectedGrid}]{
       	\vphantom{
       		$\begin{bsmallmatrix}
       		\vphantom{-1}0 \\ 
       		& 0 & 1 & 0\\ 
       		& 0 & 0 & 1 \\ 
       		& 0 & 0 & 0 \\ 
       		&  &  &  & 0 & 1 & 0 & 0 & 0\\ 
       		&  &  &  & 0 & 0 & 1 & 0 & 0\\ 
       		&  &  &  & 0 & 0 & 0 & 1 & 0\\ 
       		&  &  &  & 0 & 0 & 0 & 0 & 1\\ 
       		\vphantom{-1}&  &  &  & 0 & 0 & 0 & 0 & 0
       		\end{bsmallmatrix}$}
       	\includegraphics[width=0.3\linewidth,valign=c]{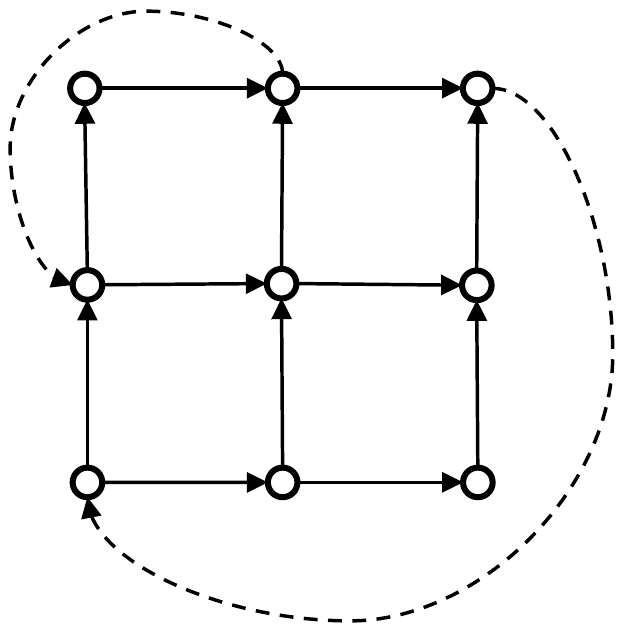}
    }
	\qquad
	\subfloat[b][\label{subfig:DirectedGridJordan}]{
          \vphantom{
		\includegraphics[width=0.3\linewidth,valign=c]{DirectedGrid}}
		$\begin{bsmallmatrix}
		\vphantom{-1}0 \\ 
		& 0 & 1 & 0\\ 
		& 0 & 0 & 1 \\ 
		& 0 & 0 & 0 \\ 
		&  &  &  & 0 & 1 & 0 & 0 & 0\\ 
		&  &  &  & 0 & 0 & 1 & 0 & 0\\ 
		&  &  &  & 0 & 0 & 0 & 1 & 0\\ 
		&  &  &  & 0 & 0 & 0 & 0 & 1\\ 
		\vphantom{-1}&  &  &  & 0 & 0 & 0 & 0 & 0
		\end{bsmallmatrix}$
	}
	\caption{The directed grid in (a) is an example of a directed acyclic
		graph with the JNF shown in (b). One possibility to destroy all Jordan blocks is adding both dashed edges $(8,4)$ and $(9,1)$.}\label{fig:DirectedGrid}
\end{figure}

Note that the common way of adding boundaries, if the two-dimensional DFT is used for spectral analysis, makes the graph a torus, which implies six added edges in this case.

\mypar{Weighted graphs} We concentrate in our theoretical
considerations on \emph{unweighted} directed graphs. This is justified
since from~\cite[Thm.~4.23]{Hershkowitz:1993a} it follows that if an
unweighted digraph is diagonalizable, then a generic weighted version
of the digraph, with weights not equal to zero, is diagonalizable as
well. Indeed, consider any weighted version of the example in
Fig.~\ref{fig:DFTModel} with nonzero weights $w_1, \dots, w_n$. Then
\begin{equation}
    \label{eq:WeightedPathGraph}
    \begin{bsmallmatrix}
	0 & w_1 & 0 & \cdots & 0 \\
	0 & 0 & w_2 & \cdots & 0 \\
	\vdots & & &  \ddots & \vdots \\
	0 & 0 & 0 & \cdots & w_{n-1} \\
	0 & 0 & 0 & \cdots & 0 \\
    \end{bsmallmatrix}
\end{equation}
has the JNF with one block shown in Fig.~\ref{subfig:shift0} with base change $V = \diag(1, \tfrac{1}{w_1}, \tfrac{1}{w_1 w_2}, \cdots, \tfrac{1}{w_1
  \dots w_{n-1}})$. On the other hand, any weighted version of the directed cycle 
\begin{equation}
    \label{eq:WeightedDirectedCycle}
    \begin{bsmallmatrix}
        0 & w_1 & 0 & \cdots & 0 \\
        0 & 0 & w_2 & \cdots & 0 \\
        \vdots & & &  \ddots & \vdots \\
        0 & 0 & 0 & \cdots & w_{n-1} \\
        w_n & 0 & 0 & \cdots & 0 \\
    \end{bsmallmatrix}    
\end{equation}
with $w_n\neq 0$ is diagonalizable. 

If the weights happen to be not generic, which can happen, for example, if they are integer values, it is straightforward to generalize Algorithm~\ref{algo:AbstractFixingAlgorithm} to this situation.

\subsection{Destroying Jordan Blocks of Directed Laplacians}
\label{subsec:AddingEdgesToKillJordanBlocks}%

We briefly explain the straightforward extension of our approach to directed Laplacians $L = D - A$, where $D$ is the matrix of outdegrees (alternatively indegrees) and $A$ the adjacency matrix. Note that this definition is not compatible with self-loops, which are thus disallowed. 

The only needed modification of Algorithm~\ref{algo:AbstractFixingAlgorithm} is to ensure that adding an edge maintains the Laplacian structure. This means that, in addition, $1$ has to be added on the main diagonal (or subtracted if an edge is removed). Thus, the perturbation $B$ has now two entries, -1 and 1, but in the same row, so the rank is still 1 and Corollary~\ref{corollary:DestroyConditionEasy} can be applied.

Note that necessarily 0 is an eigenvalue of every Laplacian with eigenvector $(1,1,\dots,1)^T$. It is also known that all Jordan blocks for eigenvalue $\lambda = 0$ have size 1 \cite{Caughman:2006}. We establish that the larger blocks (and thus for $\lambda \neq 0$) can indeed be destroyed by adding edges.

\begin{theorem}\label{prop:AddingOneEdgeToDestroyOneBlock}%
	Adding or deleting one edge is sufficient to destroy the largest
	Jordan block of a Laplacian for a chosen eigenvalue $\lambda\neq 0$.
\end{theorem}
\begin{proof}
Let $u_1,\ldots,u_r$ and $v_1, \ldots, v_r$ be the left and right eigenvectors
of the largest Jordan blocks for the eigenvalue $\lambda\neq 0$,
respectively. Then, as in Theorem~\ref{prop:AddingOneEdgeToDestroyOneBlock}, \eqref{eq:GenericityConditionEasy} yields
\begin{equation}\label{eq:lcond}
0 \neq \left(\sum_{k=1}^r v_k^T \tensor u_k^T\right) \vec(B) = w^T\vec(B).             
\end{equation}
where $B$ has the structure explained above. Assume this is not possible. Then, in particular, the first $n$ elements of $w$ are all equal, say equal to $a\neq 0$. From \eqref{eq:lcond} we get
$$
a\cdot(1,1,\dots,1)^T = v_{1,1}u_1^T + \dots v_{r,1}u_r^T,
$$
which implies that $(1,1,\dots,1)^T$ is an eigenvector for $\lambda\neq 0$ a contradiction.
\end{proof}

\mypar{Example} We consider one example.
\begin{example}
    \label{example:DirectedLaplacianDirectedGrid}%
    Consider the directed Laplacian of the graph in
    Fig.~\ref{fig:DirectedGrid}. The Jordan structure is
    \begin{equation}
        \label{eq:DirectedLaplacianDirectedGridJordan}
        \begin{bsmallmatrix}
            0 & \\ 
            & 1 & 1 & \\ 
            & 0 & 1 & \\ 
            &  &  & 1 & 1 & \\ 
            &  &  & 0 & 1 &  \\ 
            &  &  &  &  & 2 & \\ 
            &  &  &  &  &  & 2 & 1 & 0 \\ 
            &  &  &  &  &  & 0 & 2 & 1 \\ 
            &  &  &  &  &  & 0 & 0 & 2
        \end{bsmallmatrix}.
    \end{equation}
    The conditions to destroy the Jordan blocks of size greater than
    one are
    \begin{equation}
        \label{eq:DirectedLaplacianDirectedGridDestroyJordanBlocksConditions}
        \begin{split}
            b_{7,1} - b_{7,4} + b_{8,1} - b_{8,4} + b_{9,1} - b_{9,4}
            &\not= 0, \\
            b_{3,1} - b_{3,2} + b_{6,1} - b_{6,2} + b_{9,1} - b_{9,2}
            &\not= 0, \\
            2 b_{9,1} - 2 b_{9,2} - 2 b_{9,4} + 2 b_{9,5} &\not= 0.
        \end{split}
    \end{equation}
    Hence the largest Jordan blocks to the eigenvalues $1$ and $2$ can
    be destroyed by adding the edge $(9,1)$. Note that this does not
    destroy both Jordan blocks to the eigenvalue $1$, since a
    perturbation of rank one can at most destroy one Jordan block to
    an eigenvalue. The condition to destroy the remaining Jordan block
    to the eigenvalue $1$ reads
    \begin{equation}
        \label{eq:DirectedLaplacianDirectedGridDestroyJordanBlocksConditionTwo}
        \tfrac{1}{2}( b_{3,2} + b_{3,4} + b_{6,2} + b_{6,4} + b_{7,2}
        - b_{7,4} + b_{8,2} - b_{8,4}) \not= 0.
    \end{equation}
Thus the choices are the same four edges as in \eqref{eq:DirectedGridSecondBlockCondition} plus four additional edges.
\end{example}

\mypar{Adjacency matrix versus Laplacian} In general, the diagonalizability of the adjacency matrix or Laplacian are different properties. Fig.~\ref{fig:CounterexamplesLaplacianDiagAdjacency} shows counterexamples for both implications.
\begin{figure}
    \centering
    \hfill
    \subfloat[\label{subfig:AdjacencyButNotLaplacian}]{
      \vphantom{
      $
      \begin{bsmallmatrix}
          \tfrac{1}{2}(5 - \sqrt{5}) \\
          & 0 \\
          & & 1 \\
          &   &   & \tfrac{1}{2}(5 + \sqrt{5})
      \end{bsmallmatrix}
      $}
      \includegraphics[width=0.075\textwidth,valign=c]{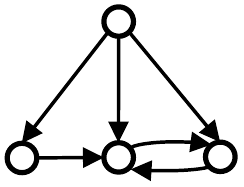}
    }
    \hfill
    \subfloat[\label{subfig:AdjacencyButNotLaplacianAdjJord}]{
      $
      \begin{bsmallmatrix}
          -1  \vphantom{\tfrac{1}{2}}\\
          & 0 & 1 \\
          & 0 & 0 \\
          &   &   & 1 \vphantom{\tfrac{1}{2}}
      \end{bsmallmatrix}
      $
    }
    \hfill
    \subfloat[\label{subfig:AdjacencyButNotLaplacianLapJord}]{
      $
      \begin{bsmallmatrix}
          \tfrac{1}{2}(5 - \sqrt{5}) \\
          & 0 \\
          & & 1 \\
          &   &   & \tfrac{1}{2}(5 + \sqrt{5})
      \end{bsmallmatrix}
      $
    }
    \hfill
    \

    \hfill
    \subfloat[\label{subfig:LaplacianButNotAdjacency}]{
    {
      \includegraphics[width=0.2\linewidth,valign=c]{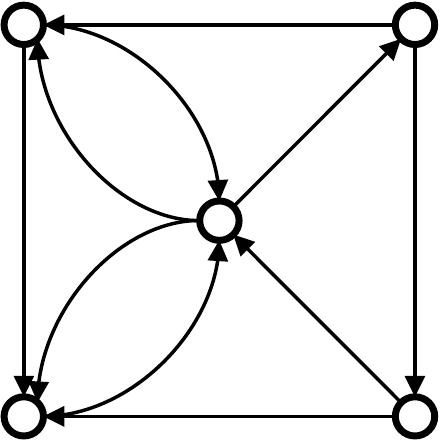}
      }
    } 
    \hfill
    \subfloat[\label{subfig:LaplacianButNotAdjacencyAdjJord}]{
      \vphantom{
      \includegraphics[width=0.2\linewidth,valign=c]{NonDiagLaplacianButAdjacency}}
      $
      \begin{bsmallmatrix}
          \tfrac{\I}{2} (\I - \sqrt{3}) \\
          &-1 \\
          & & 0 \\
          & & & 2 \\
          & & & & \tfrac{\I}{2} (\I + \sqrt{3})
      \end{bsmallmatrix}
      $
    } 
    \hfill
    \subfloat[\label{subfig:LaplacianButNotAdjacencyLapJord}]{
      \vphantom{
      \includegraphics[width=0.2\linewidth,valign=c]{NonDiagLaplacianButAdjacency}}
      $
      \begin{bsmallmatrix}
          0 \vphantom{\tfrac{\I}{2}}\\
          & 2 & 1\\
          & & 2 & 1 \\
          & & & 2 \\
          & & & & 4 \vphantom{\tfrac{\I}{2}}
      \end{bsmallmatrix}
      $
    }     
    \caption{Counterexamples to implications of diagonalizability of
      adjacency/Laplacian. From left to right: graph, adjacency Jordan
      structure, Laplacian Jordan structure.}
    \label{fig:CounterexamplesLaplacianDiagAdjacency}
\end{figure}


\section{Algorithm and Implementation}
\label{sec:NumericalAlgorithm}%

In this section we explain how to implement the Algorithm~\ref{algo:AbstractFixingAlgorithm} numerically. The challenge is to achieve both numerical stability and scalability to large graphs, where the former is necessary for the latter. More concretely, we address two main challenges. First, the algorithm in Fig.~\ref{algo:AbstractFixingAlgorithm} requires the eigenvectors of all largest Jordan blocks, but the Jordan basis is not computable for larger graphs. Second, small numerical errors can lead to the addition of unnecessary edges. Thus we need suitable heuristics.

Finally, we argue that in real-world graphs most of the non-trivial Jordan blocks are associated with the eigenvalue $0$. We exploit this observation with a special algorithm variant that enables scaling to graphs with several thousands of nodes.

We implemented our algorithms in Matlab\footnote{The code is available as open source at \url{https://github.com/bseifert-HSA/digraphSP-generalized-boundaries}.}, which requires some additional details that we explain as well.

\begin{figure*}\footnotesize
	\begin{algorithmic}
		\Function{\textbf{DestroyJordanBlocks}}{$A$}
		\State $U, V \leftarrow$ left and right eigenvectors of $A$
		\While {$\operatorname{rank}(V,\epsilon_R) < n$}
		\Comment{\parbox[t]{.55\linewidth}{Check if eigenvectors form a basis}} 
		\State $D \leftarrow \operatorname{acos}(| V^T \cdot V|)$ 
		\Comment{\parbox[t]{.55\linewidth}{Pairwise angles between
				subspaces spanned by eigenvectors}}
		\State $k \leftarrow \operatorname{argmax}_k( \#( D_{k,i} <
		\epsilon_D))$ 
		\Comment{\parbox[t]{.55\linewidth}{Index of eigenvector for largest Jordan block}}
		\State $(i,j) \leftarrow
		\operatorname{argmax}_{i,j}(|U_{i,k}| \cdot 
		|V_{j,k}|)$ s.t. $A_{i,j} = 0$
		\Comment{\parbox[t]{.55\linewidth}{Choose edge which destroys
				the largest Jordan block}}
		\State $A_{i,j} \leftarrow 1$
		\Comment{\parbox[t]{.55\linewidth}{Add the new edge}}
		\State $U,V \leftarrow$ right and left eigenvectors of $A$
		\EndWhile
		\State \textbf{return} $A$
		\EndFunction 
	\end{algorithmic}
	\caption{Algorithm for obtaining a digraph with a diagonalizable
		adjacency matrix by destroying all Jordan blocks. In our
		experiments a rank tolerance $\epsilon_R = 10^{-6}$ and an
		eigenspace angle tolerance $\epsilon_D$ of one degree were used.}
	\label{algo:FixDigraphSpectrum}%
\end{figure*}

\begin{figure*}\footnotesize
	\centering
	\begin{algorithmic}
		\Function{\textbf{DestroyZeroEigenvalues}}{$A$}
		\State $D \leftarrow$ eigenvalues of $A$
		\While{there exists $|D_i| < \epsilon_Z$}
		\Comment{\parbox[t]{.55\linewidth}{Very small eigenvalues are considered zero}}
		\State $u,v \leftarrow$ right/left eigenvector to $D_i$          
		\Comment{\parbox[t]{.55\linewidth}{Compute only one right/left eigenvector to $D_i$}}
		\State $(i,j) \leftarrow \operatorname{argmax}_{i,j}(|u_i| \cdot
		|v_{j}|)$ s.t. $A_{i,j} = 0$
		\Comment{\parbox[t]{.55\linewidth}{Choose edge to be added}}
		\State $A_{i,j} \leftarrow 1$
		\Comment{\parbox[t]{.55\linewidth}{Add the new edge}}
		\State $D \leftarrow$ eigenvalues of $A$
		\EndWhile
		\State \textbf{return} $A$
		\EndFunction 
	\end{algorithmic}    
	\caption{Algorithm for removing all zero eigenvalues of a
		digraph. In our experiments we choose $\epsilon_Z = 10^{-3}$ as tolerance for
		identifying zeros.}
	\label{algo:KillAllZeros}
\end{figure*}

\subsection{Numerical Algorithm: Details}
\label{subsec:ObstaclesNumerical}%

We explain the additional details to make the algorithm
in Fig.~\ref{algo:AbstractFixingAlgorithm} efficiently applicable in practice.

\mypar{Eigenvectors of largest Jordan blocks} The algorithm requires the eigenvectors of all largest Jordan blocks. In practice, we cannot determine the largest Jordan blocks via computing the JNF (except for very small graphs). Typical implementations of the
eigendecomposition, as the one used in Matlab, give as output for
non-diagonalizable matrices still a complete matrix of eigenvectors, in which, however, each eigenvector is usually repeated as often as the corresponding Jordan block size. Thus, as a first heuristic, we compute the pairwise angles between the spaces
spanned by the eigenvectors and determining the largest group with
angles very close to zero.

Since there is no certain way to obtain the left and right eigenvectors of all largest blocks as needed by Corollary~\ref{corollary:DestroyConditionEasy}, we only compute them for one block. Thus, as a second heuristic, we aim to make only one summand in condition~\eqref{eq:GenericityConditionEasy} nonzero, which makes the entire sum nonzero in the generic case that no cancellation occurs. Further, doing so still guarantees that one Jordan block (but not necessarily the largest) is destroyed:
  \begin{lemma}[\cite{Savchenko:2020a}]
      \label{lemma:SoundHeuristics}%
      If $u^T B v \not= 0$ for some left and right eigenvectors
      associated with an eigenvalue $\lambda$ of $M$, then exactly one
      Jordan block to $\lambda$ will be destroyed under the transition
      from $M$ to $M + B$.
  \end{lemma}

Both heuristics are robust in the sense that in the worst case they add useless edges, which does not affect termination, as discussed before.

\mypar{Choice of edge} In general,
Algorithm~\ref{algo:AbstractFixingAlgorithm} produces in each
iteration several choices for the edge to add, based on the sparsity
pattern of our chosen (with above heuristic) $u$ and $v$. Since very small nonzero values could be rounding errors, we choose the edge corresponding to the maximal
absolute value in both $v$ and $u$ for stability.

\mypar{Sparsity and eigenvalue \boldmath$0$} An adjacency matrix can
have nontrivial Jordan blocks for any eigenvalue (this can be shown
using the rooted product of graphs). However, in real-world graphs and some
of the random graph models commonly considered, we frequently observe
the eigenvalue 0 with high multiplicity, which was also observed
in~\cite{Deri.Moura:2017b}. This observation can be explained
with Theorem~\ref{theorem:CoefficientTheoremDigraphs}: real-world
graphs are typically sparse so it is likely that several edges are not
part of any cycle, which yields a large factor $x^m$ in the
characteristic polynomial.

This observation is valuable, since it is computationally much cheaper
to compute only one eigenvector to a known eigenvalue. Matlab offers
the function \texttt{eigs} for this purpose. As an additional benefit, this function also has special support for sparse matrices unlike the \texttt{eig} function.

\subsection{Implementation}
\label{subsec:ImplementationDetails}%

We used the above insights and heuristics to refine Algorithm~\ref{algo:AbstractFixingAlgorithm} into two algorithms. Algorithm~\ref{algo:FixDigraphSpectrum} destroys all Jordan blocks of a digraph to obtain a diagonalizable adjacency matrix. As an optional preprocessing step, the considerable more efficient Algorithm~\ref{algo:KillAllZeros} adds edges to remove all zero eigenvalues and hence yields an invertible adjacency matrix. Note that the
algorithms require numerical tolerance parameters to determine which
eigenvalues are 0, and which eigenvectors should be considered as
collinear or equal.

\mypar{DestroyJordanBlocks} From the above considerations we
can now derive the numerical Algorithm~\ref{algo:FixDigraphSpectrum}. First we calculate the left and right eigenmatrices $U,V$ of the adjacency matrix $A$. While $A$ is not diagonalizable, which we check by testing if $V$ is rank-deficient, we destroy iteratively the Jordan blocks. For this we use our heuristic and first calculate all the angles $D$ between the subspaces spanned by the elements of $V$.
Then we obtain the eigenvector, which most likely corresponds to the
largest Jordan block, by finding the index $k$ for which most entries of
$D$ are approximately zero. To find the edge, we maximize over the
product of the entries of the left and right eigenvectors
$|U_{i,k}| \cdot |V_{j,k}|$ under the constraint that $A_{i,j} = 0$.
The case of the random edge in Algorithm~\ref{algo:AbstractFixingAlgorithm} occurs here if $A_{i,j} = 1$ whenever the product is $>0$, i.e., if the maximum is zero.

For the implementation of Algorithm~\ref{algo:FixDigraphSpectrum} in
Matlab we use the \texttt{eig} function with the \texttt{nobalance} option.  With these options, Algorithm~\ref{algo:FixDigraphSpectrum} is applicable
to all matrix sizes for which one can calculate the complete
eigendecomposition of a full matrix.

For the computation of $\rank(V)$, Matlab requires a tolerance, for which we chose $\epsilon_R = 10^{-6}$, meaning that the smallest singular value fulfills $\sigma_{\min} > 10^{-6}$. In~\cite{Domingos.Moura:2020a} this condition was used to define a Fourier basis as numerical stable. Thus our constructed bases are stable in the same sense. As tolerance $\epsilon_D$ to identify two eigenvectors we observed that an angle of one degree is a good choice.

\mypar{DestroyZeroEigenvalues} Algorithm~\ref{algo:KillAllZeros} destroys all zero eigenvalues. Here we first calculate the eigen\emph{values} of the adjacency
matrix $A$. As long as an eigenvalue is approximately zero, we calculate an associated left and right eigenvector. Then we destroy the Jordan block to that eigenvalue similar as in Algorithm~\ref{algo:FixDigraphSpectrum}. 

We implement Algorithm~\ref{algo:KillAllZeros} using sparse matrices in the compressed sparse row (CSR) format and use the Matlab function \texttt{eigs} to find one eigenvector to the numerical eigenvalue zero and destroy the corresponding Jordan block. Thus, Algorithm~\ref{algo:KillAllZeros} scales to all matrix sizes for
which one can calculate one eigenvector for a sparse matrix. In our
experiments we consider eigenvalues as zero if their absolute value
is $\leq \epsilon_Z = 10^{-3}$.

Since we argued already that in real-world graphs typically many
Jordan blocks are associated to the eigenvalue $0$, one can obtain a
significant speedup by first removing all zeros from the eigenvalues
of a graph using Algorithm~\ref{algo:KillAllZeros} and then,
afterwards, applying the more costly
Algorithm~\ref{algo:FixDigraphSpectrum} to destroy the remaining
Jordan blocks.

\mypar{Robustness} We note that our algorithms yield an inherent
robustness property: our parameter settings ensure that the
final digraph obtained does not have very small eigenvalues, does
not have almost collinear eigenspaces, and produces a numerically
stable Fourier basis.

\mypar{Complexity}
In the implementation of DestroyZeroEigenvalues we use the sparse CSR matrix format, which requires space of size $O(\max(n,m))$, where $n$ is the size of the matrix and $m$ the number of nonzero entries~\cite{Gilbert.Moler.Schreiber:1992a}.
Updating the adjacency matrix with a new edge requires $O(m)$
operations. \texttt{eigs} calculates an eigenvector to the eigenvalue zero
in time $O(n k^2)$ using a Krylov-Schur algorithm~\cite{Stewart:2001a}; $k$ depends on the rate of convergence which is hard to estimate beforehand, 

The implementation of DestroyJordanBlocks relies on full matrices and hence
the required storage is $O(n^2)$. The complexity of computing
the complete eigendecomposition and the pairwise angles is $O(n^3)$.


\section{Experiments and Applications}
\label{sec:Application}%

We evaluate our proposed algorithm and implementation with two kinds of experiments. First, we apply our algorithm to a set of random and real-world graphs to make them diagonalizable (and possibly invertible) and investigate the results. Then we show a Wiener filter as prototypical application that is enabled by using our approach that first establishes a complete basis of eigenvectors. Finally, we consider also the case of a Laplacian to demonstrate that our approach is equally applicable.

The experiments in this section, unless stated otherwise, were performed on a computer with an
Intel Core i9-9880H CPU and 32 GB of RAM. 

\subsection{Computing generalized boundary conditions}
\label{subsection:ModificationSomeGraphs}%

\mypar{Random digraphs} 
In our first experiment we apply our algorithm DestroyJordanBlocks in Fig.~\ref{algo:FixDigraphSpectrum} to four different classes of random digraphs~\cite{Prettejohn.Berryman.McDonnell:2011a}. We briefly recall their properties.

The Erdős–Rényi model creates homogeneous digraphs in the sense that
the degree distribution of the nodes decays symmetrically from the
mean degree, the average path length increases as the graph size
increases, and its clustering coefficient reduces as the graph size
increases.

The Watts-Strogatz model leads to small-world digraphs which means they have large clustering coefficients, unlike the Erdős–Rényi random graphs.

The Barabási-Albert model yields scale-free digraphs in the sense that their degree distribution is very inhomogeneous, which means they contain a large numbers of nodes with small
degree and a only a few hubs with large degree.

The fourth model is Klemm-Eguílez, which combines the small-world property of the Watts-Strogatz model with the scale-freeness of the Barabási-Albert model. Since it is conjectured that real-world networks are scale-free and small-world, these graphs may be particularly realistic. 

For each model we generated 100 random weakly connected graphs\footnote{If a graph is not weakly connected the components can be processed separately.} with 500 nodes. We set the model parameters to obtain an average of about 5000 edges in each case. For the Erdős–Rényi model we choose a success probability of connecting two nodes of $0.02$. We created Watts-Strogatz model graphs with $10$ edges to each node in the
initial ring lattice and a rewiring probability of $0.001$. The
Barabási-Albert model got as parameters a seed size of $10$ and an
average degree of the nodes of $10$. Finally we used the
Klemm-Eguílez model with seed size $5$ and a probability of
connecting to non-active nodes of $0.1$. The parameters are explained in~\cite{Prettejohn.Berryman.McDonnell:2011a}. 

For Erdős–Rényi, all generated graphs were diagonalizable. For the other models we 
summarize the results of applying DestroyJordanBlocks in Table.~\ref{tab:ResultsRandomGraphs}. The table reports the minimum, median, and maximum number of edges added to make them diagonalizable and the runtime to do so. The first main observation is that our algorithm works in each case and with a runtime that is easily acceptable for a one-time preprocessing step. For Watts-Strogatz very few edges are sufficient in all cases, whereas for the other two up to 1\% additional edges may be needed in the worst case.

\begin{table}
    \centering
    \begin{tabular}{@{}lllllllll@{}}\toprule
      & \multicolumn{2}{c}{min}& &
        \multicolumn{2}{c}{median}& &
        \multicolumn{2}{c}{max} \\
      \cmidrule{2-3} \cmidrule{5-6} \cmidrule{8-9}
      & edges & time & & edges & time & & edges & time \\
      \midrule
      Watts-Strogatz & 0 & 0.2s & & 1 & 0.5s & & 3 & 1.3s \\
      Barabási-Albert & 36 & 4.4s & & 44 & 10s & & 55 & 31s \\
      Klemm-Eguílez & 10 & 2.2s & & 27 & 6s & & 47 & 9s \\
      \bottomrule          
    \end{tabular}
    \caption{Edges added and runtime of DestroyJordanBlocks for three different
      random graph models with 500 nodes and approximately 5000
      edges.}
    \label{tab:ResultsRandomGraphs}
\end{table}

Next, we consider three real-world graphs.

\mypar{USA graph} First, we consider a small digraph consisting of the 48
contiguous US states with edges going from lower to higher latitude (see Fig.~\ref{fig:fixedUSAGraph}) that has been a popular use case in several publications (e.g.,~\cite{Shafipour.Khodabakhsh.Mateos.Nikolova:2019a,Furutani.Shibahara.Akiyama.Hato.Aida:2019a}). The graph consists of 48 nodes and 105 edges and is an extreme case since it is acyclic, i.e., only has the eigenvalue 0, with 7 Jordan blocks of sizes $13, 10, 9, 5, 5, 4, 2$, respectively. Application of DestroyJordanBlocks yields (the minimal needed number of) $7$ added edges shown in Fig.~\ref{fig:fixedUSAGraph} dashed in blue. The eigenvalues of the modified graph are shown in Fig.~\ref{fig:USAGraphEigenvalues}. They are all simple eigenvalues, and well-separated, which is ideal for any subsequent GSP analysis. Further, Fig.~\ref{fig:USAGraphEigenstructure} shows the angles between (spaces generated by the) eigenvectors. On the left for the Jordan basis of the original USA graph (which for this size is still computable) and on the right for the eigenbasis of the modified graph. The basis is not far from orthogonal, a property that will become more pronounced for the larger graphs considered next.

\begin{figure} \centering
    \includegraphics[width=0.8\linewidth]{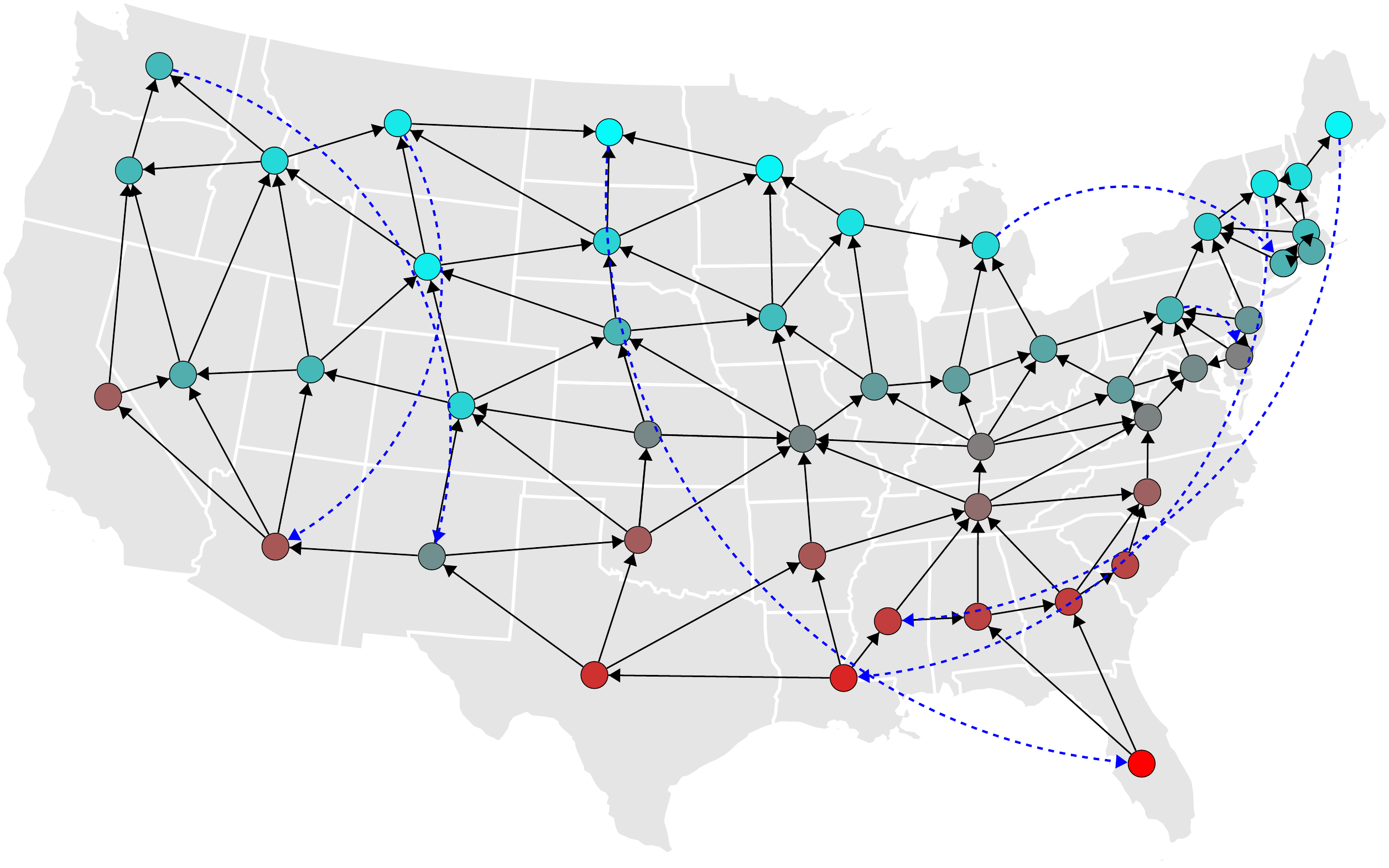}
    \caption{The USA graph. The 7 new edges added by
      DestroyJordanBlocks are shown as dashed blue.}
    \label{fig:fixedUSAGraph}
\end{figure}
\begin{figure}
    \centering
    \includegraphics[width=0.7\linewidth]{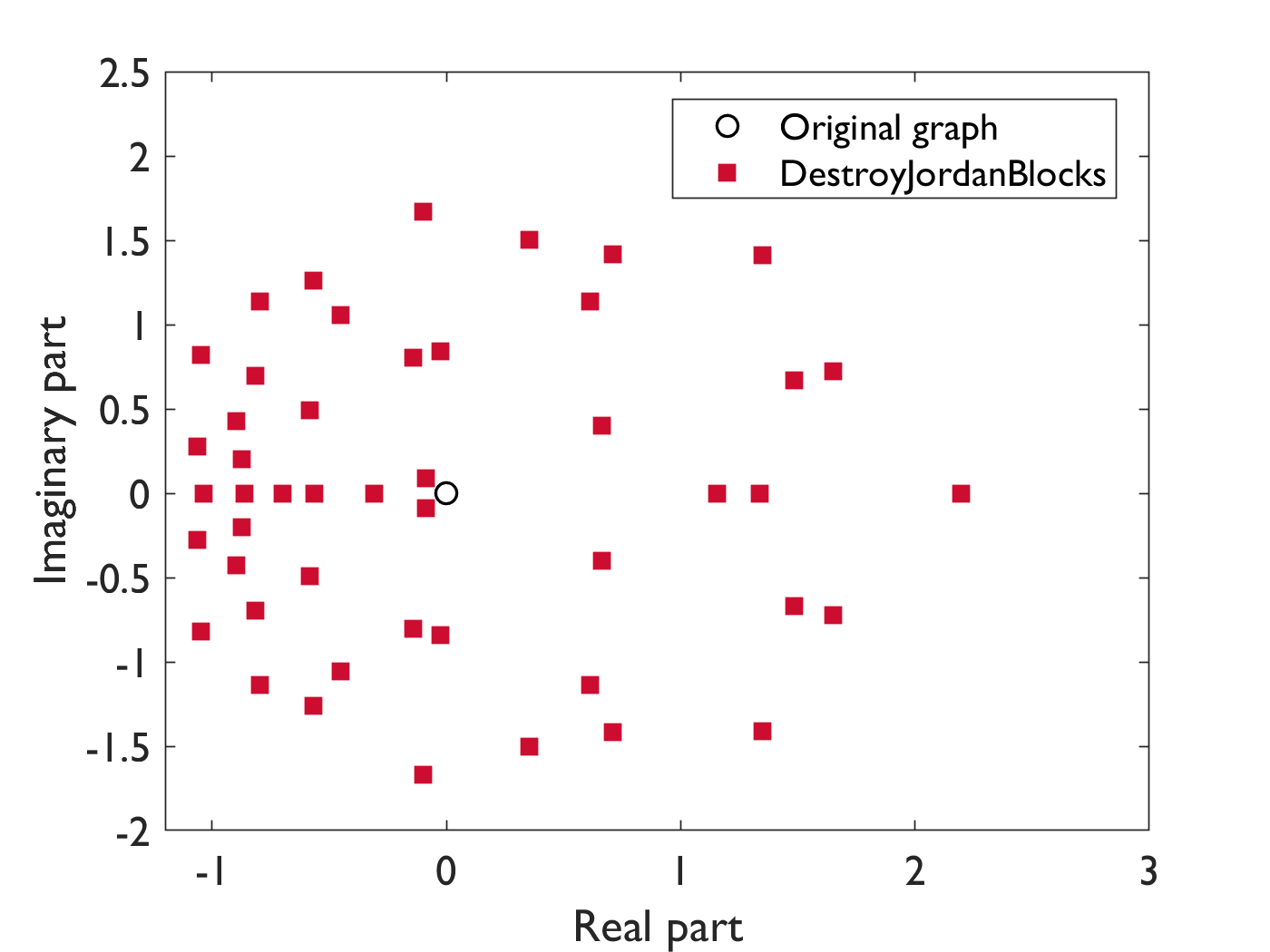}
    \caption{The eigenvalues of the USA graph (black circle),
      and after making it diagonalizable (red squares).}
    \label{fig:USAGraphEigenvalues}
\end{figure}
\begin{figure}
    \centering
    \includegraphics[width=0.49\linewidth]{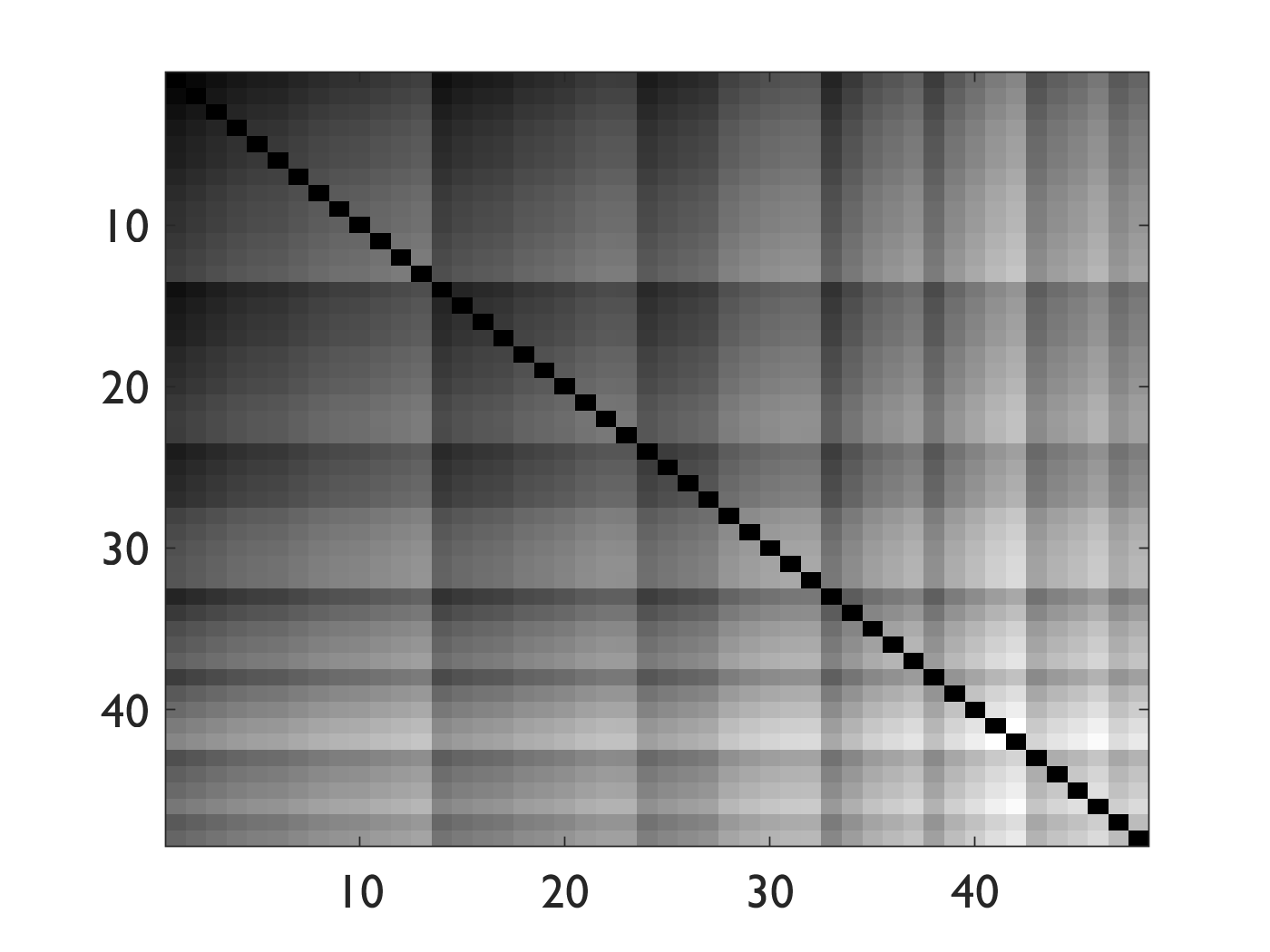}
    \includegraphics[width=0.49\linewidth]{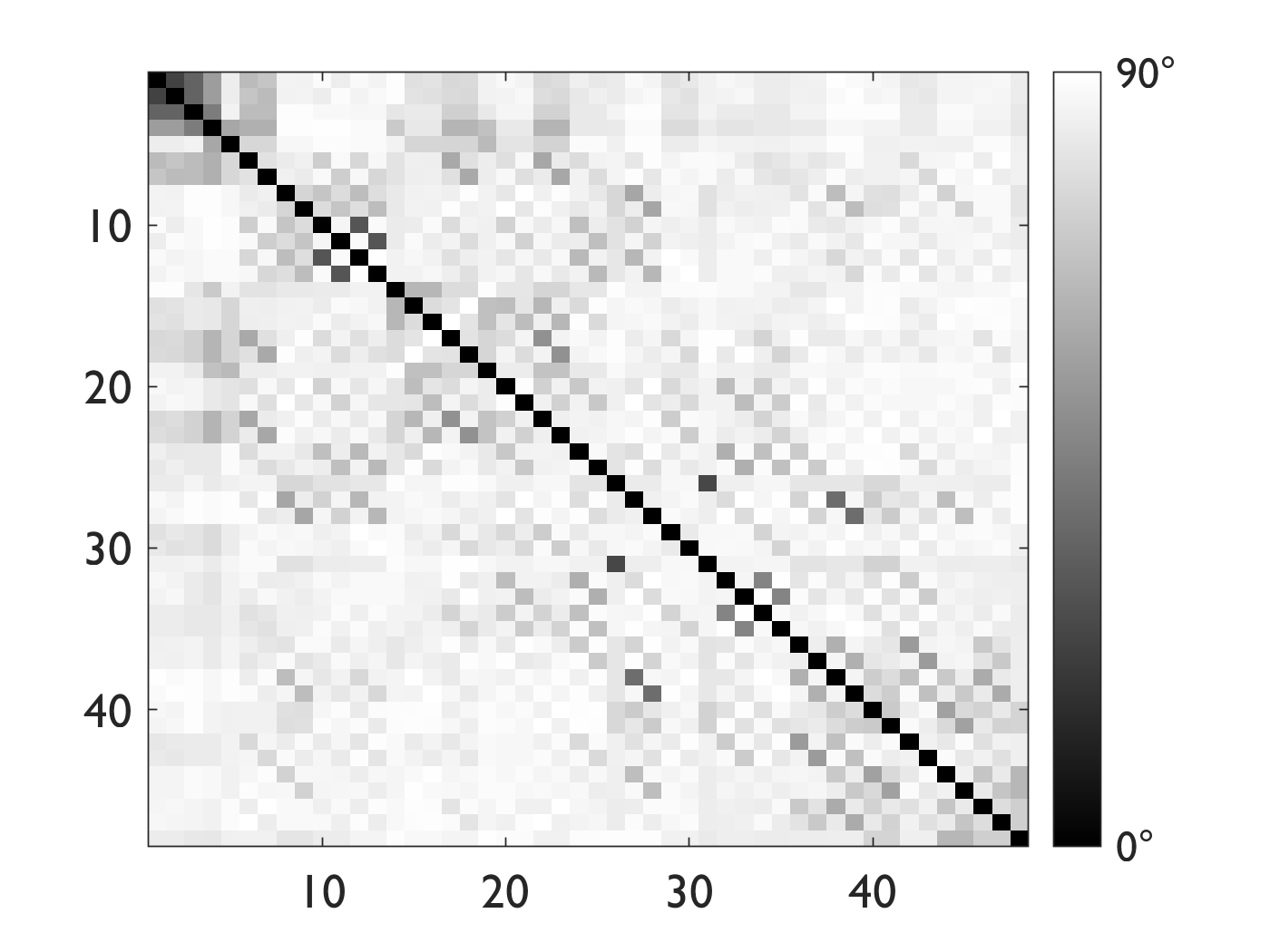}
    \caption{The angles between the computed generalized eigenvectors
      of the original USA graph (left) and the angles between the
      computed eigenvectors of the USA graph with the additional
      edges (right).}
    \label{fig:USAGraphEigenstructure}
\end{figure}

\mypar{Manhattan taxi graph} Next we demonstrate that our
algorithm can process large-scale graphs that are particular
challenging in numerical stability. First we consider the
Manhattan taxi graph used in~\cite{Li.Moura:2020a,Domingos.Moura:2020a}\footnote{The graph
  and the graph signal is based on data available at \url{https://www1.nyc.gov/site/tlc/about/tlc-trip-record-data.page}} and shown in Fig.~\ref{fig:ManhattanGraph}. The graph
consists of 5464 nodes, each representing a spatial location at the
intersection of streets or along a street. A directed edge means that traffic is allowed to directly move from one node to the other. The total number of edges in the graph is 11568.

\begin{figure}
    \centering
    \includegraphics[width=\linewidth]{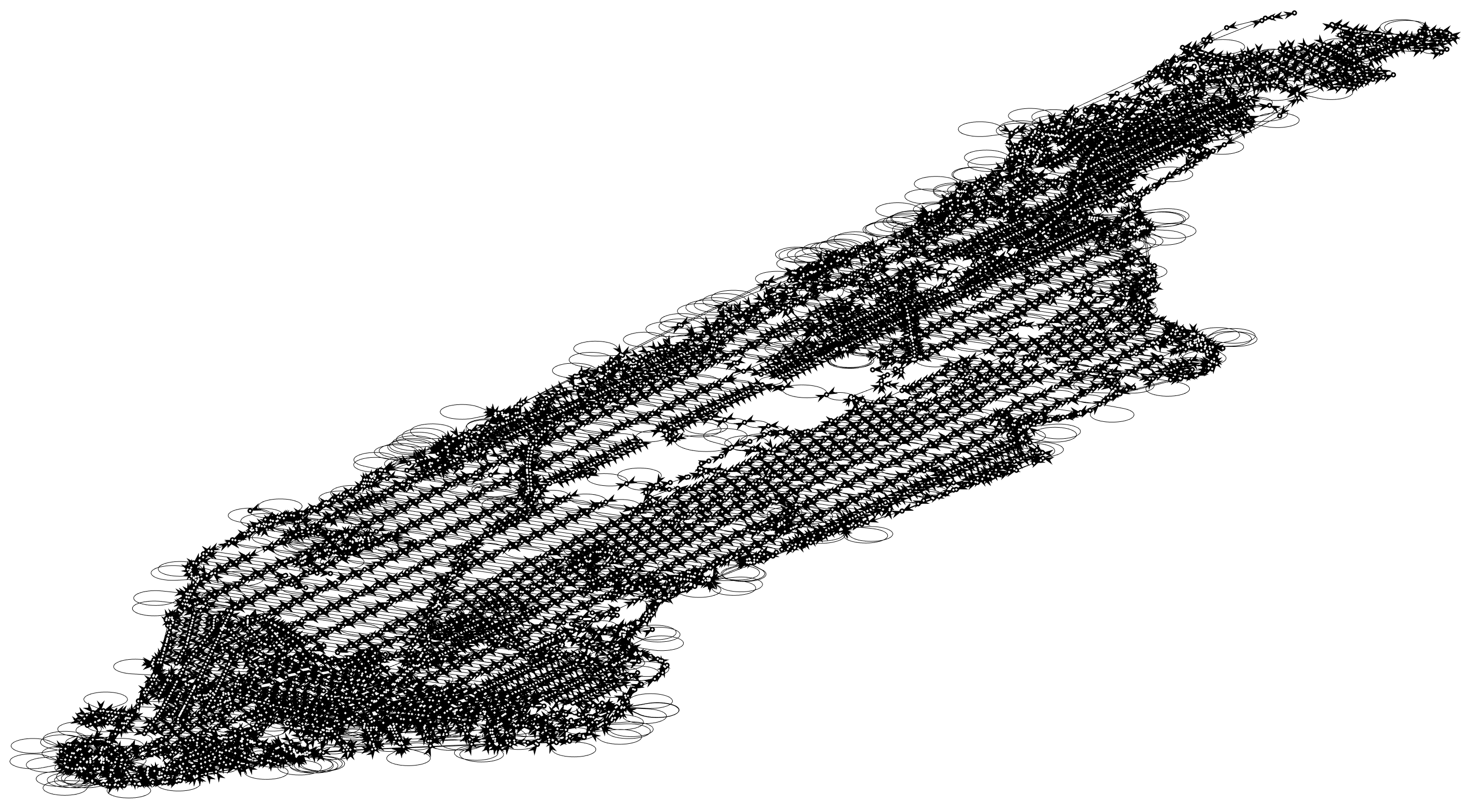}
    \caption{The Manhattan graph with 5464 nodes.}
    \label{fig:ManhattanGraph}
\end{figure}

Because of the large scale and rank deficiency, as explained in
Section~\ref{subsec:ImplementationDetails}, we first apply
DestroyZeroEigenvalues in Fig.~\ref{algo:KillAllZeros} to first remove
all zero eigenvalues, which took 2.3 minutes and added 772 edges
(about 6.7\%). Then we applied DestroyJordanBlocks in
Fig.~\ref{algo:FixDigraphSpectrum}, which added another 1 edge in 2.5
minutes for a total processing time of about 5 minutes. Our algorithm
guarantees that the resulting graph has a computed eigenmatrix with
full rank (with tolerance
$\sigma_\text{min}\geq \epsilon_R = 10^{-6}$), and a minimal angle
between computed eigenvectors of $\epsilon_D \geq$ one degree.
However, Fig.~\ref{fig:ManhattanHistogram} shows that the eigenbasis
is even very close orthogonal. It is also numerically stable:
$\sigma_{\min} = 0.0017$, $\sigma_{\max} = 4.9158$, i.e., the
condition number is $\kappa = \sigma_{\max} / \sigma_{\min} = 2892$, which
means we can compute a valid Fourier transform by
inversion\footnote{If the condition number is $\kappa$, about $\log_{10} \kappa$ decimal digits of precision are lost when inverting the matrix in floating point~\cite[p.~95]{Trefethen.Bau:1997a}. Here it is only 3 digits out of 16 available in double precision.}.

\begin{figure}
    \centering
    \includegraphics[width=0.49\linewidth]{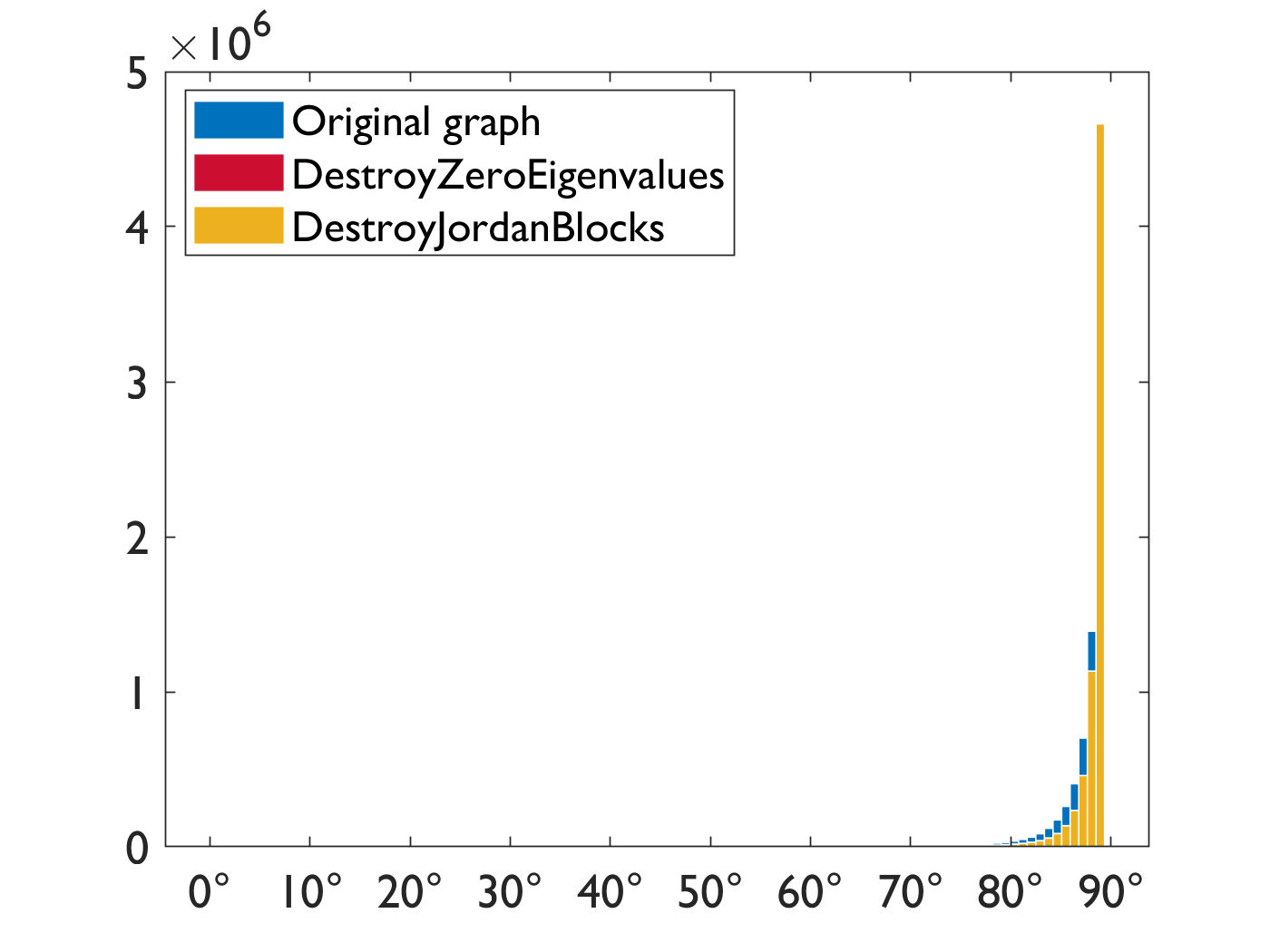}
    \includegraphics[width=0.49\linewidth]{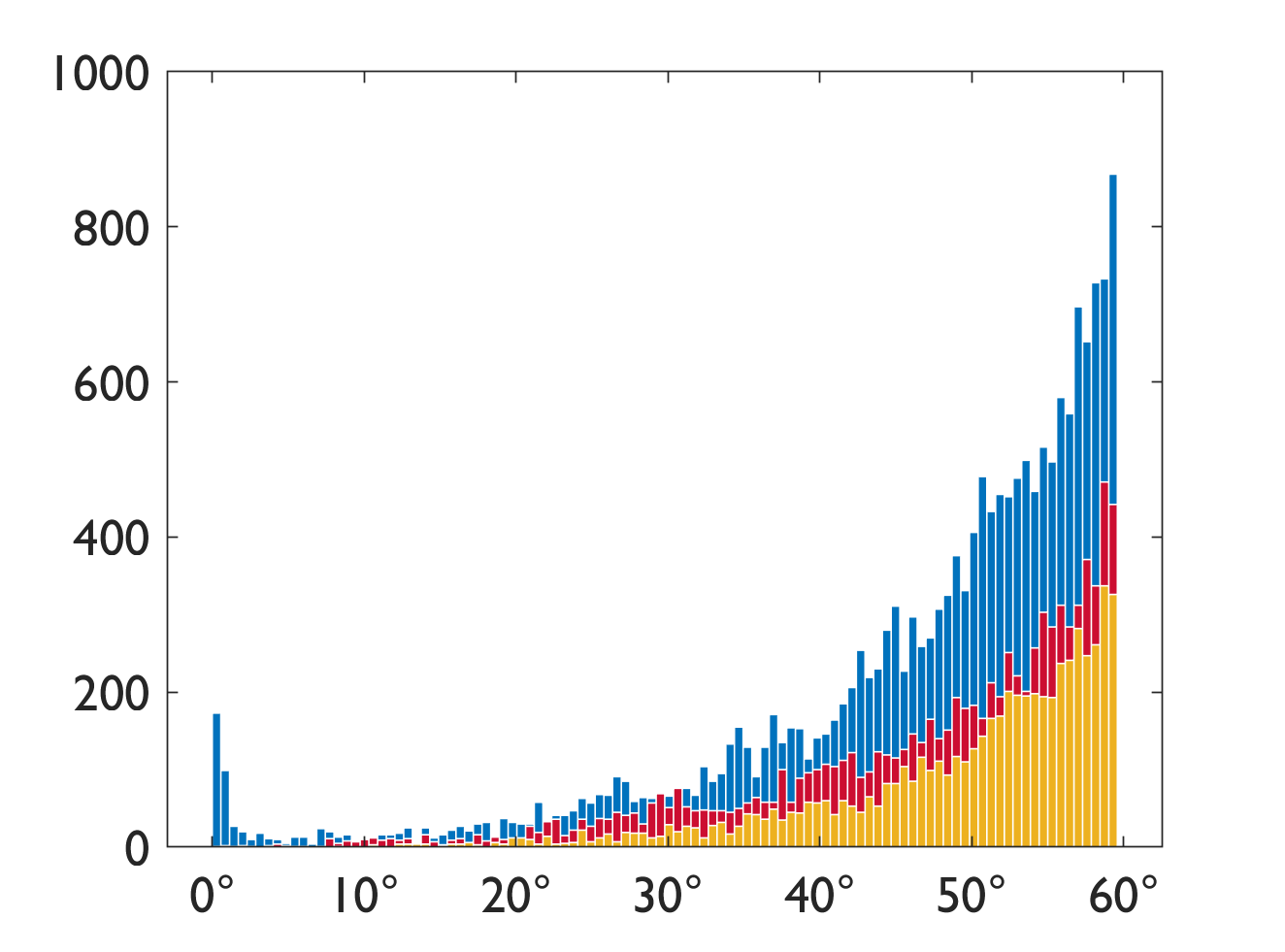}
    \caption{Manhattan graph: The histogram of all $5464^2$ angles between the spaces spanned by
      the computed eigenvectors: for all angles (left), and zoomed in on angles $\leq 60$ degrees (right).}
    \label{fig:ManhattanHistogram}
\end{figure}

To show the gain in computational complexity, we also applied only
DestroyJordanBlocks from Fig.~\ref{algo:FixDigraphSpectrum}. For this
experiment we used a computer with Intel Xeon CPU E5-2660, 2.2 GHz
with 128 GB of RAM. The algorithm added 243 additional edges within 19
hours. Even though the number of edges added differs significantly, the
eigenspace angle distribution for both approaches turned out very similar, meaning in
both cases one obtains an almost orthogonal graph Fourier transform.
The eigenbasis obtained by applying only
DestroyJordanBlocks is slightly less stable:
$\sigma_{\min} = 0.0006$, $\sigma_{\max} = 4.897$ for a
condition number of $\kappa = 7802$. Note that when using only DestroyJordanBlocks, the adjacency matrix still has eigenvalue zero with a high (namely 536) multiplicity. 

In the following we consider only the previous modified graph obtained with the fast method that combines both algorithms.

Fig.~\ref{fig:ManhattanGraphEigenvaluesCompare} shows that the eigenvalues of the modified graph lie in a similar range as those of the original graph, except for those near zero that we destroyed.

\begin{figure*}
    \centering
    \subfloat[\label{fig:ManhattanGraphEigenvaluesCompare} Eigenvalues of $A$ (blue) and $A + B$ (red).]{
      \includegraphics[width=0.33\linewidth]{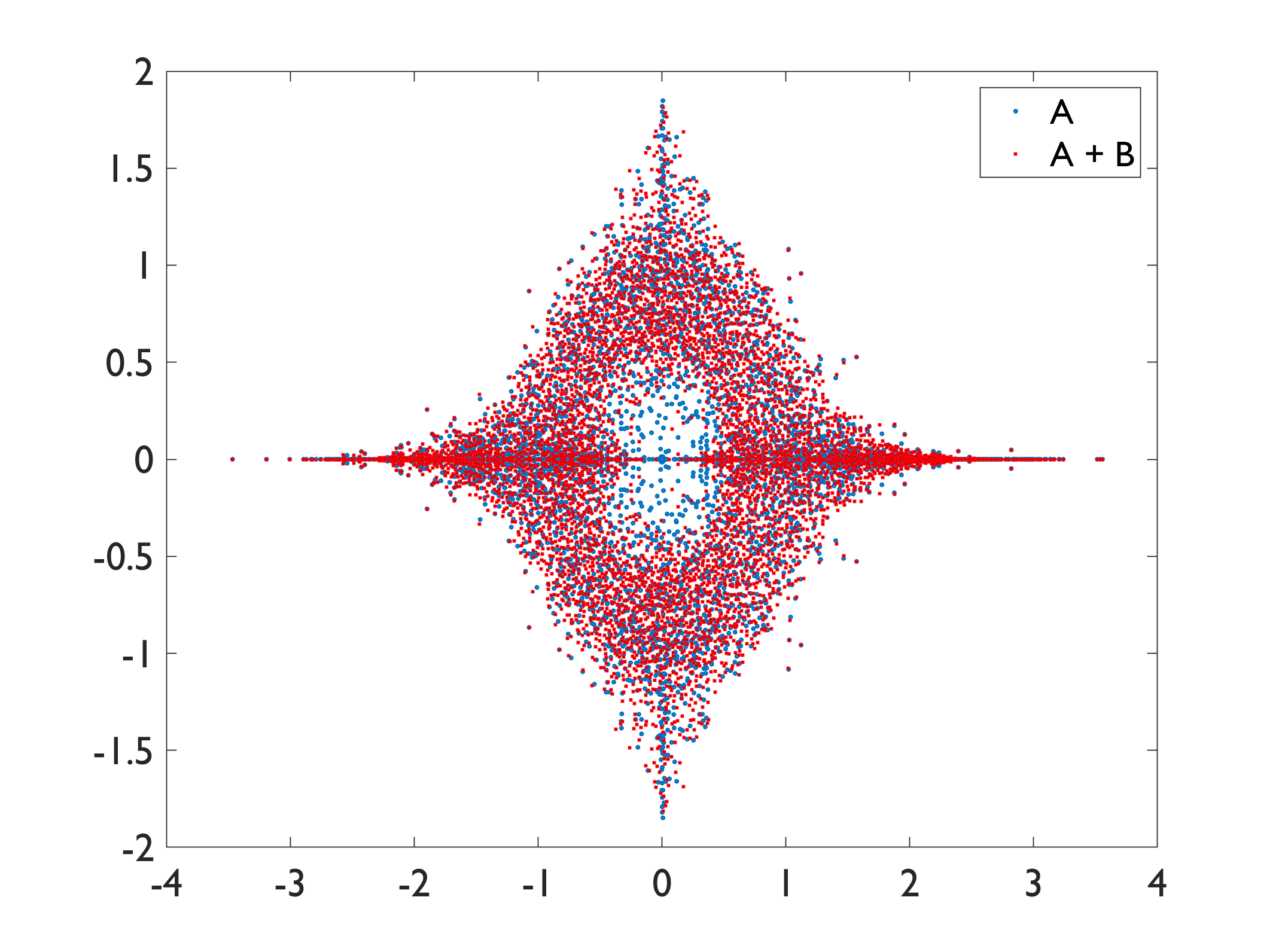} 
    }
    \subfloat[$\TV_A(v)$ vs.~$\TV_{A+B}(v)$.
    \label{fig:TotalVariationManhattanCompare}]{
      \includegraphics[width=0.33\linewidth]{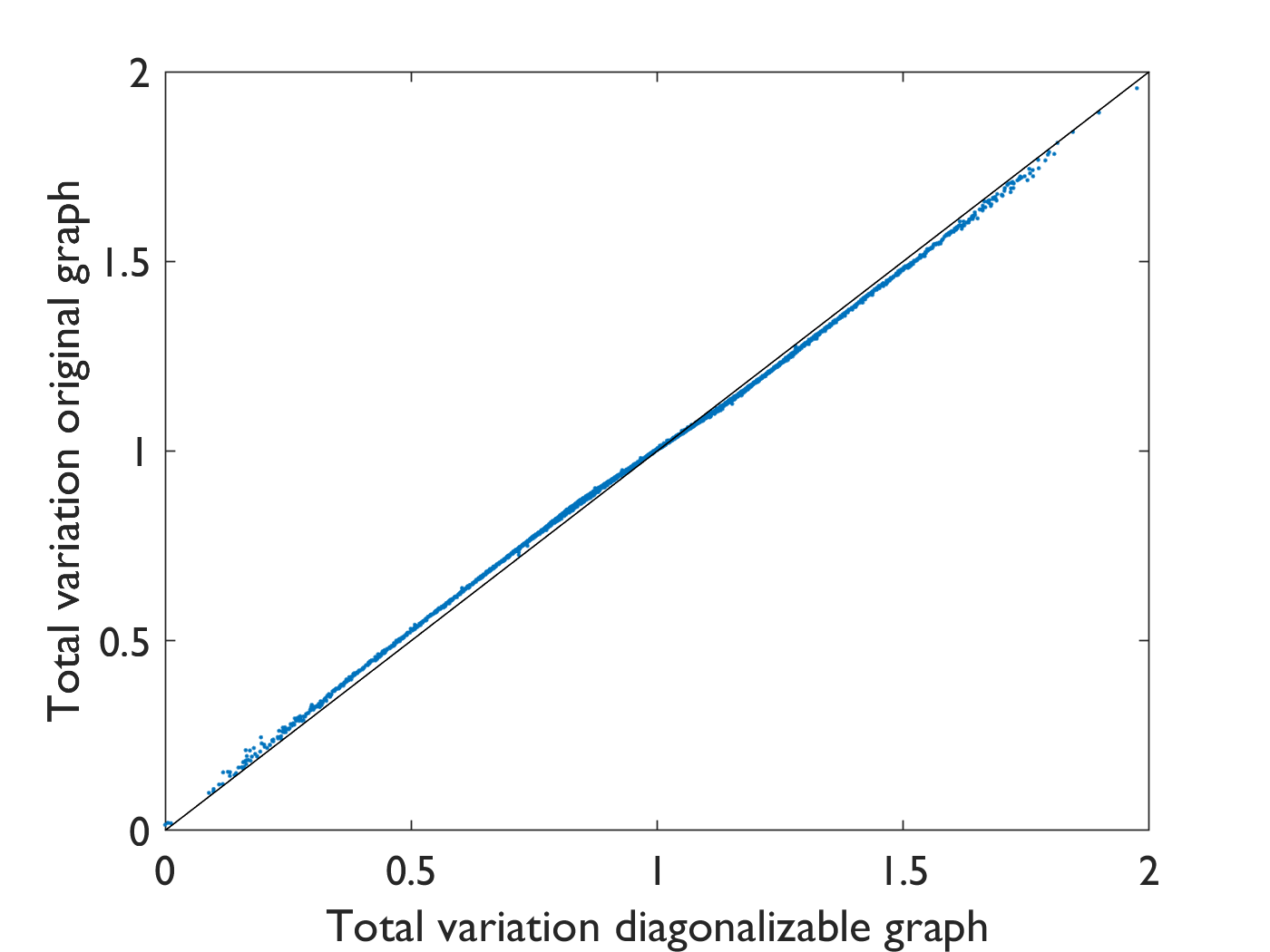} 
    } 
    \subfloat[Distribution of $\TV_A(v)$. \label{fig:TotalVariationPlotManhattan}]{
      \includegraphics[width=0.33\linewidth]{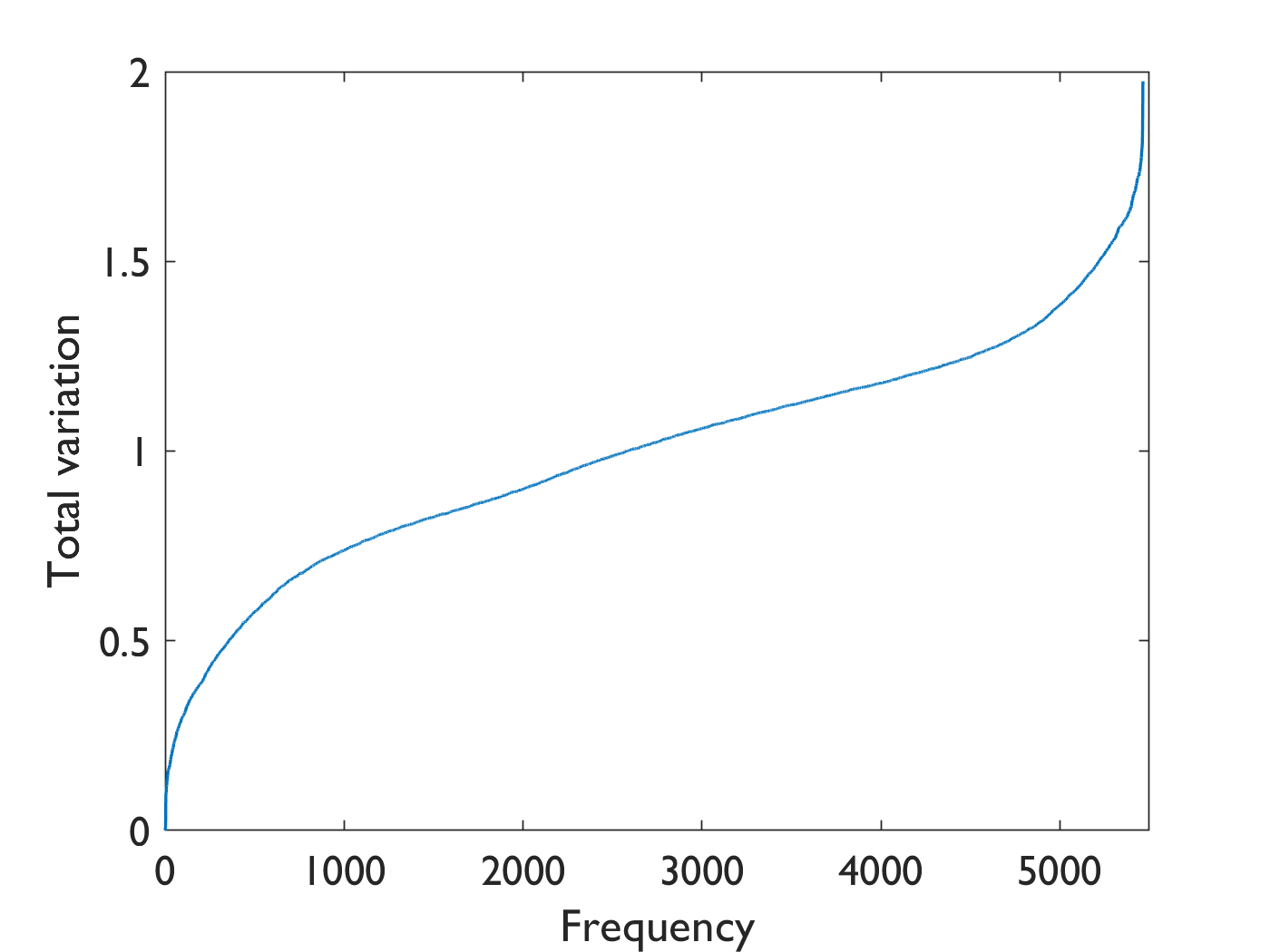}
    }
    \caption{Properties of the eigenvalues and eigenvectors $v$ of the
      modified Manhattan graph $A + B$.}
    \label{fig:SpectralPropertiesManhattanModified}
\end{figure*}

Basis vectors $v$ can be ordered by total variation $\TV_A$ (see \eqref{eq:GraphTotalVariation}) w.r.t.~the adjacency matrix $A$. For an eigenvector $v$ with $||v||_1 = 1$ and eigenvalue $\lambda$, $\TV_A(v) = |1-\lambda/\lambda_\text{max}|$. We noted earlier (Lemma~\ref{thm:NewEigenvectorsForOldGraph}) that the eigenvectors $v$ of our modified graph ($A+B$) are, in a sense, approximate eigenvectors for the original $A$. Here we compare the total variations $\TV_{A+B}(v)$ and $\TV_A(v)$ of the eigenbasis of $A+B$ when used as basis for $A$.

Fig.~\ref{fig:TotalVariationManhattanCompare} plots $\TV_A(v)$ against
$\TV_{A+B}(v)$. Even though about 6.5\% of edges were added, the total variations are
almost equal. This means that our method preserved the ordering of frequencies and thus the notion of low and high frequency. Thus, for example, a low-pass filter designed for the diagonalizable graph $A+B$ will be a low pass filter for the original $A$.

We also show the distribution curve of the total variations in Fig.~\ref{fig:TotalVariationPlotManhattan}. Interestingly, it is similar to the curve in~\cite[Fig.~16]{Domingos.Moura:2020a}, even though a very different approximation method was used.

The work in~\cite{Li.Moura:2020a} used as signal on the Manhattan graph the number
of hourly taxi rides starting from each node. Similarly, we take the number of taxi rides averaged over the first half of 2016. The obtained graph signal is shown in
Fig.~\ref{fig:ManhattanGraphSignal}. 

\begin{figure}
    \centering
    \includegraphics[width=0.7\linewidth]{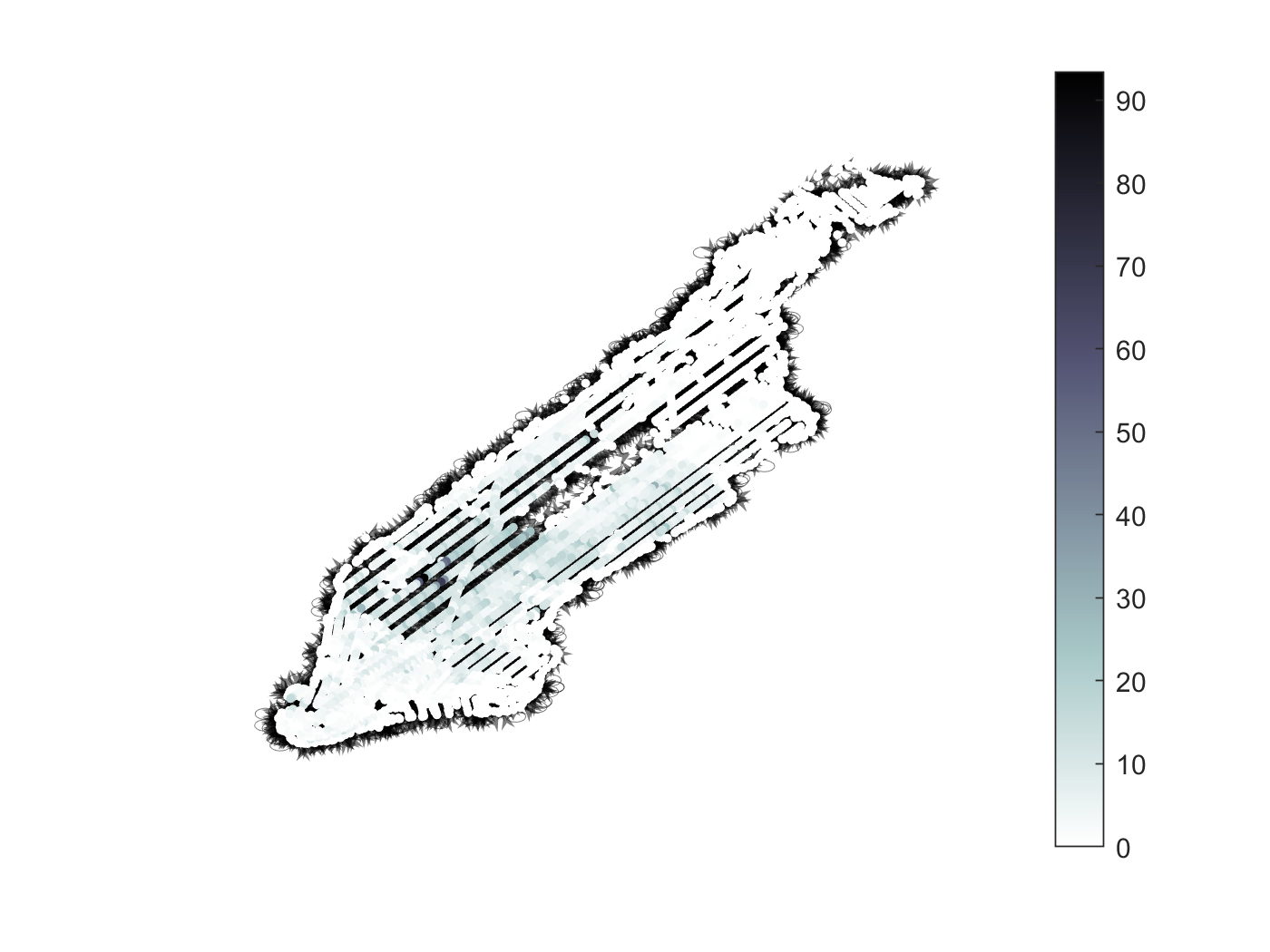}
    \caption{The graph signal of averaged hourly taxi rides on the
      Manhattan graph.}
    \label{fig:ManhattanGraphSignal}
\end{figure}

We apply the graph Fourier transform for the graph we obtained by
applying DestroyZeroEigenvalues and DestroyJordanBlocks to the
Manhattan graph to this signal. For this, we sorted the eigenvectors with
respect to total variation and normalized them to $\norm{v}_2 = 1$. 
The obtained magnitude signal spectrum is shown in
Fig.~\ref{fig:ManhattanGraphSignalSpectrum}. 
\begin{figure}
    \centering
    \includegraphics[width=0.7\linewidth]{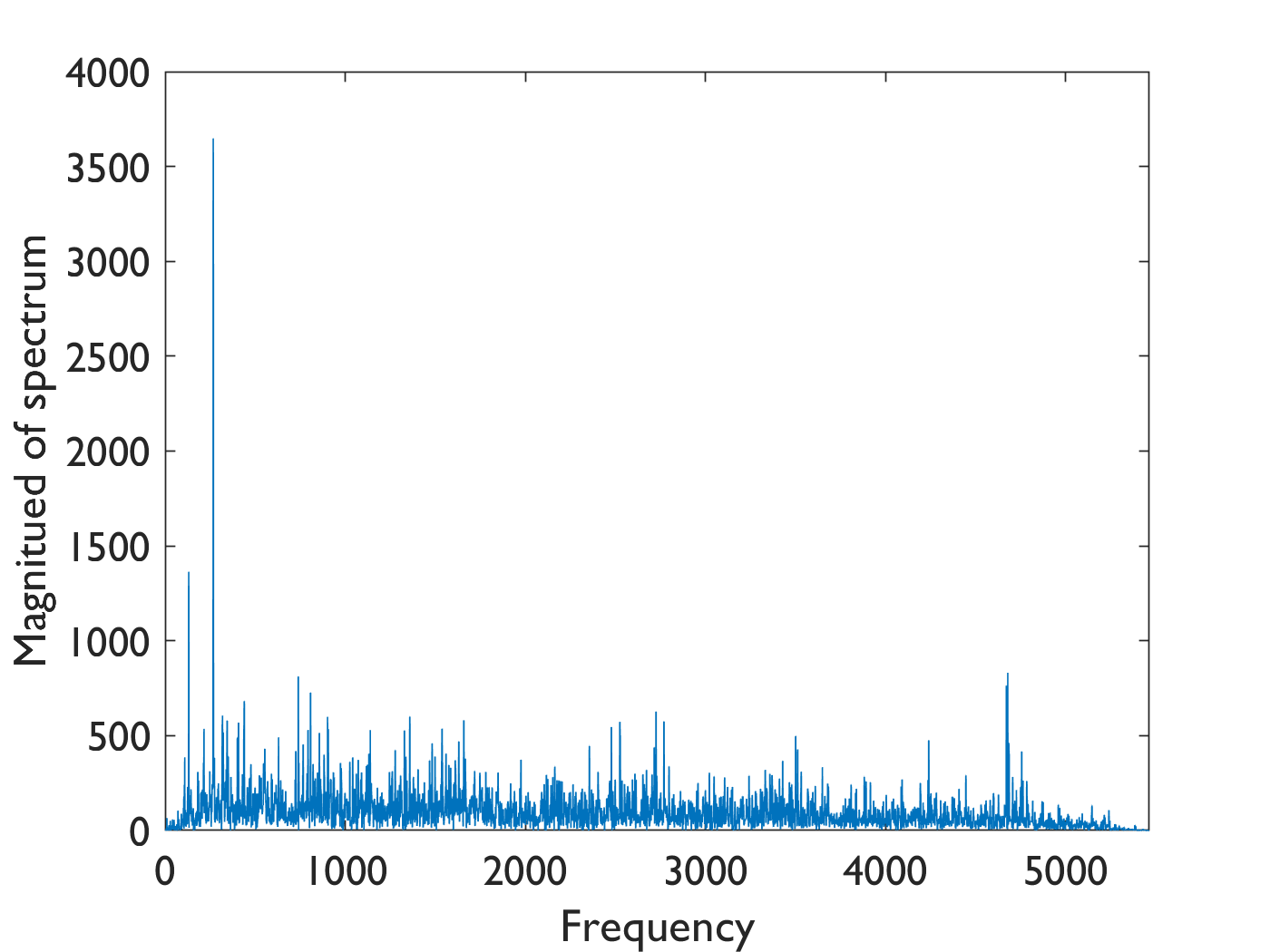}  
    \caption{The spectrum of the Manhattan graph signal. The
      frequencies are ordered by total variation.}
    \label{fig:ManhattanGraphSignalSpectrum}
\end{figure}

\mypar{Citation graph} As second large real-world graph we use the arXiv HEP-PH
citation graph released in~\cite{Gehrke.Ginsparg.Kleinberg:2003a}\footnote{The graph is
  available online at \url{https://snap.stanford.edu/data/cit-HepPh.html}}. For our
experiments we used a weakly connected subgraph with 4989 vertices and
17840 edges shown in Fig.~\ref{fig:CitationGraph}, where the nodes are vertically placed by publication time.

As a citation graph this graph is very close to being acyclic and thus has almost all eigenvalues 0, i.e., it is particularly challenging.

\begin{figure}
    \centering
    \includegraphics[width=\linewidth]{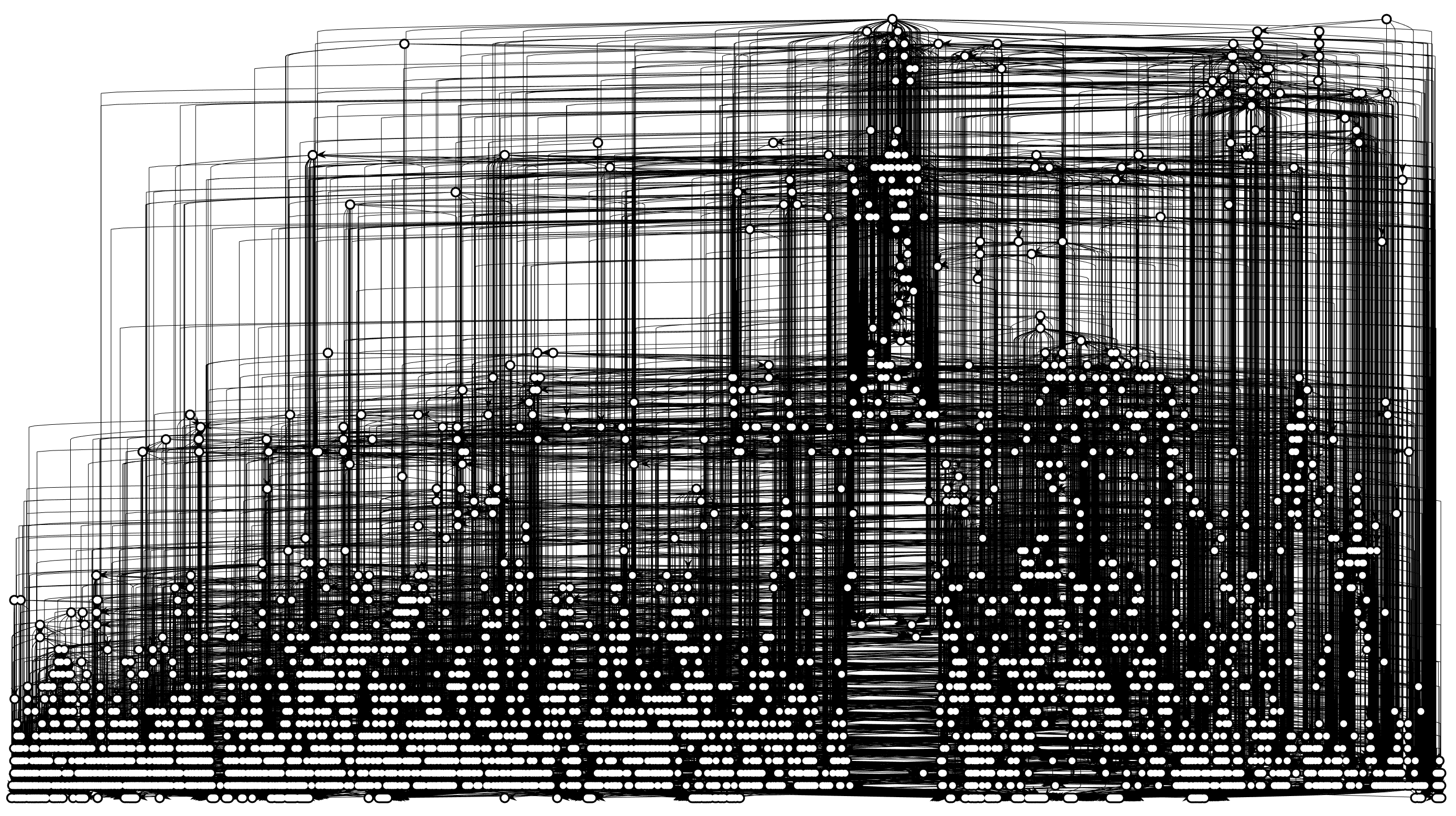}
    \caption{A citation graph with 4989 nodes.}
    \label{fig:CitationGraph}
\end{figure}

As before, we first apply DestroyZeroEigenvalues in
Fig.~\ref{algo:KillAllZeros} to remove all eigenvalue zeroes, which
took 9.5 minutes and added 1890 edges (about 10.5\%). Using
DestroyJordanBlocks then added another 21 edges in 22 minutes and
gives the usual guarantees on the minimal angle between the computed
eigenspaces. However, Fig.~\ref{fig:CitationHistogram} shows that, as
for the Manhattan graph, the eigenbasis is even almost orthogonal. The
obtained Fourier basis, with $\norm{v}_2 = 1$, is again numerically stable: $\sigma_{\min} = 0.0026$, $\sigma_{\max} = 5.1847$, for a condition number of $\kappa = 1994.1$).

\begin{figure}
    \centering
    \includegraphics[width=0.49\linewidth]{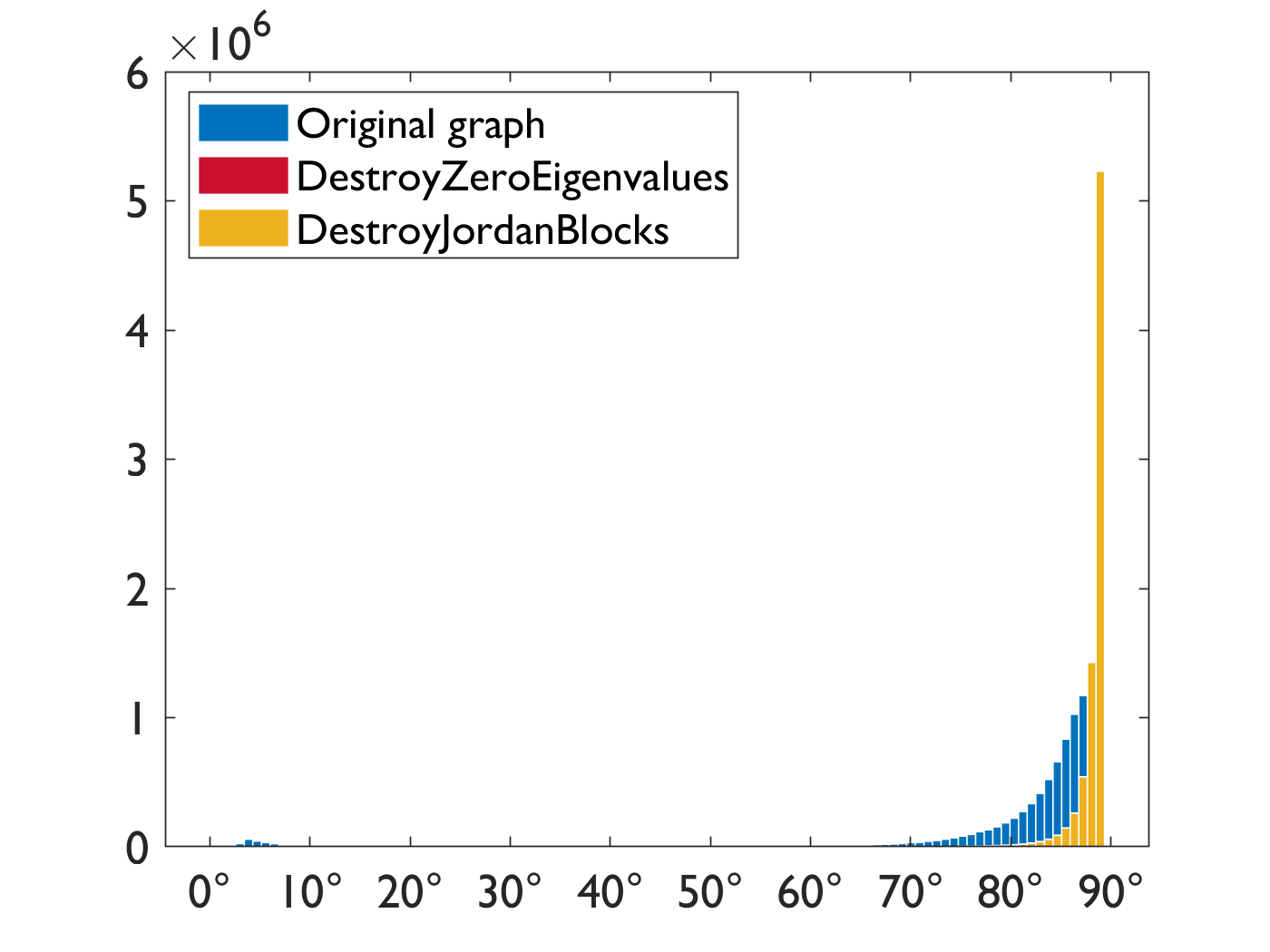}
    \includegraphics[width=0.49\linewidth]{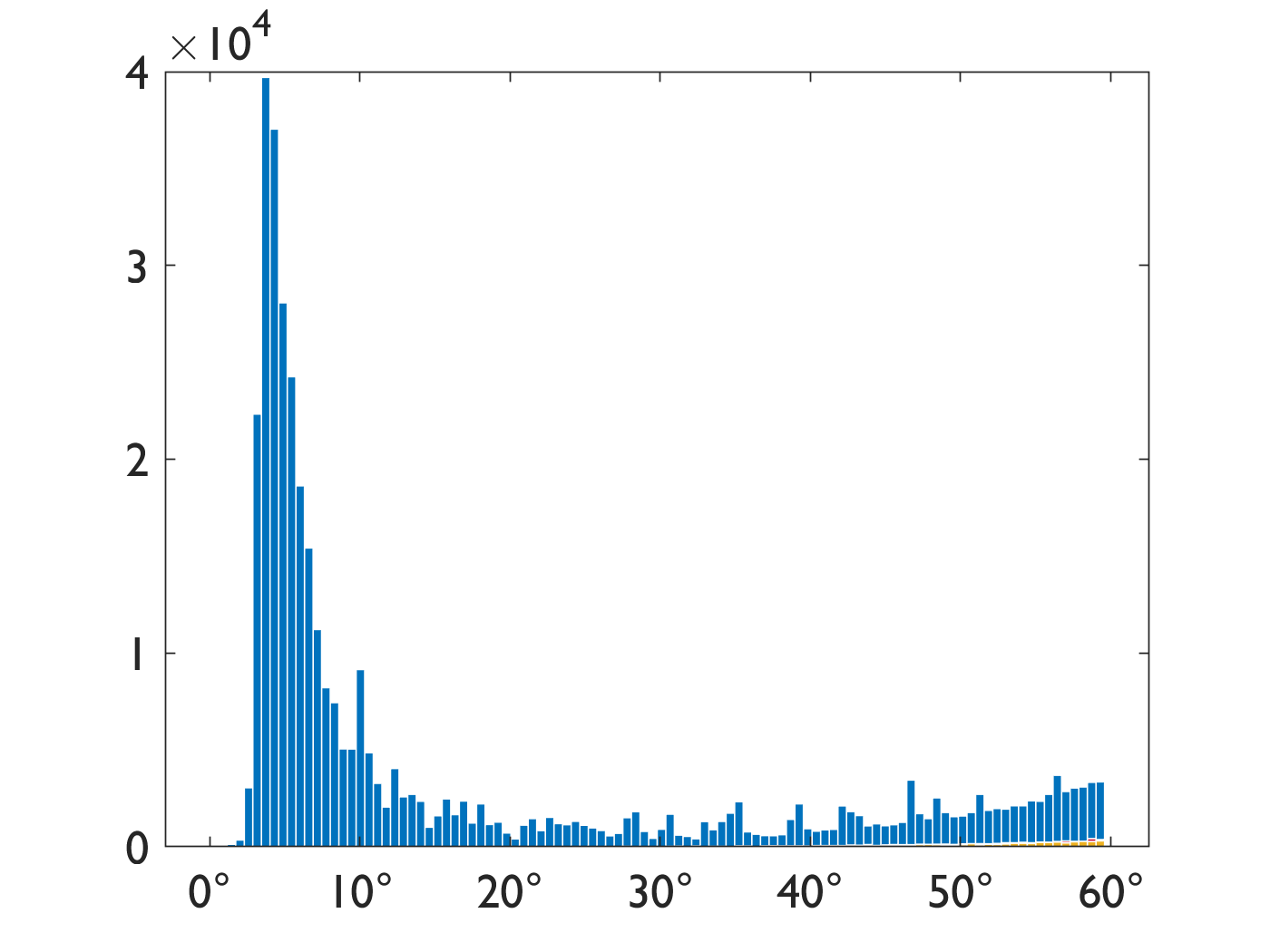}
    \caption{Citation graph: The histogram of all $4989^2$ angles between the spaces spanned by
      the computed eigenvectors for the citation graph: for all angles
      (left), and zoomed in on angles $\leq 60$ degrees (right).}
    \label{fig:CitationHistogram}
\end{figure}

As for the Manhattan graph,
Fig.~\ref{fig:TotalVariationCitationCompare} plots $\TV_A(v)$ against
$\TV_{A+B}(v)$. Even though about 10\% of edges were added, they are
close to equal and very close to order preserving.
\begin{figure}
    \centering
    \includegraphics[width=0.7\linewidth]{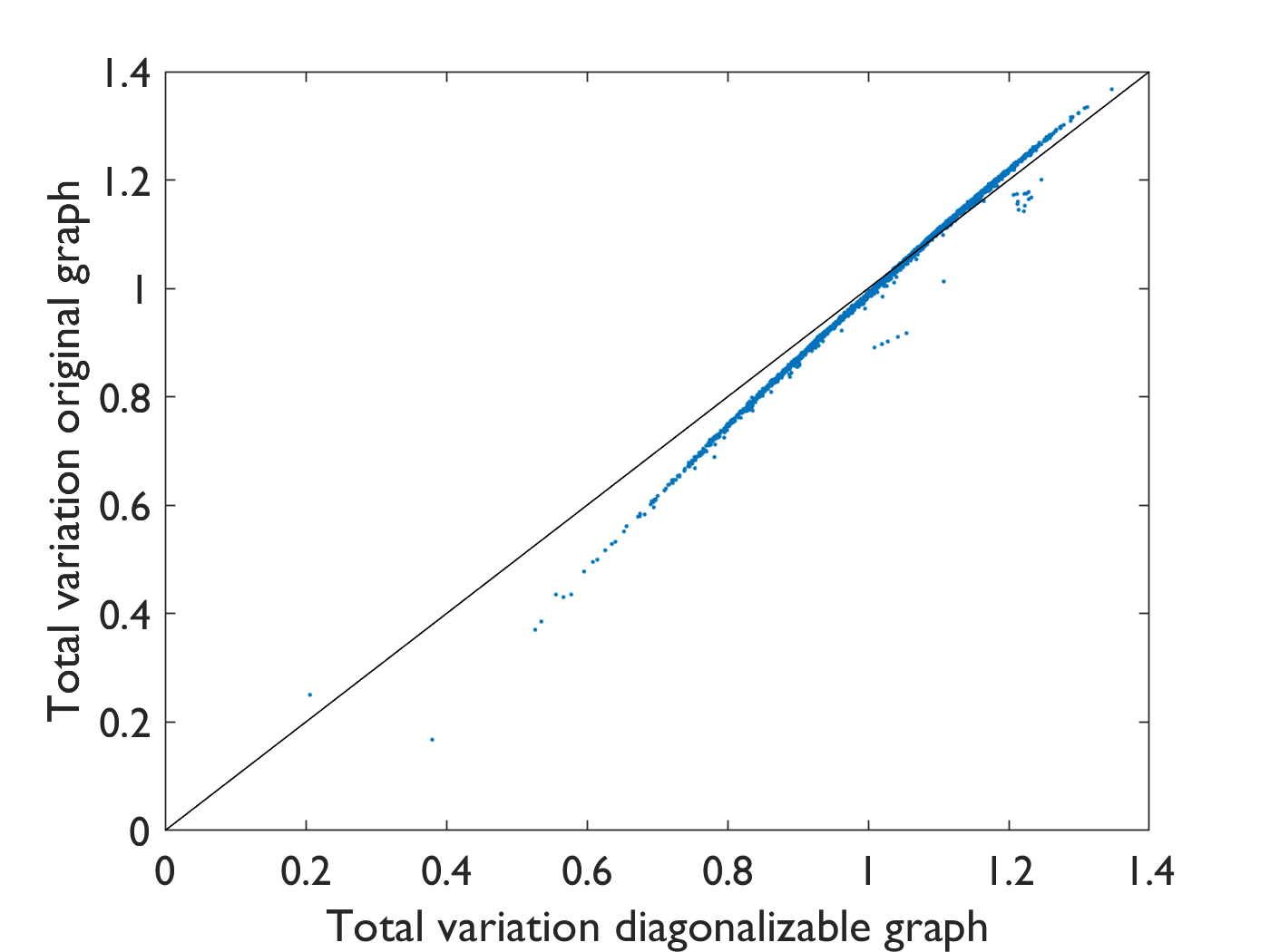}
    \caption{$\TV_A(v)$ vs.~$\TV_{A+B}(v)$ for the eigenvectors $v$ of $A+B$.}
    \label{fig:TotalVariationCitationCompare}
\end{figure}

\subsection{Wiener filtering with energy preserving shift}
\label{subsec:WienerFiltering}%

The work in~\cite{Gavili.Zahng:2017a} introduced an energy-preserving shift for graphs and digraphs but required the adjacency matrix to be diagonalizable. We show that our work can be used as a preprocessing step to establish this property to then enable further SP. As example, we use the generalization of Wiener filtering to graphs show-cased in~\cite{Gavili.Zahng:2017a}. First, we briefly provide background from \cite{Gavili.Zahng:2017a}.

\mypar{Energy-preserving shift} Let $A = V D V^{-1}$ with $D$ diagonal and let
$\Lambda_e = \diag(\lambda_{e_1}, \dots, \lambda_{e_N})$, with
$ \lambda_{e_k} = \E^{-2\I \pi (k-1)/n}$. The energy-preserving graph
shift is then defined as
\begin{equation}
    \label{eq:EnergyPreservingGraphShift}
    A_e = V \Lambda_e V^{-1}.
\end{equation}
Thus, $\norm{\Fourier s} = \norm{\Fourier (A_e s)}$ and $n$ applications to a signal reproduce the original: $A_e^n x = x$. If the eigenvalues of $A$ are all simple, $A_e$ is a polynomial in $A$, i.e., a filter.

\mypar{Graph Wiener filter} Consider a graph signal $x$ and a noisy
measurement of the signal $y = x + n$. The graph Wiener filter of
order $L$ has the form
\begin{equation}
    \label{eq:WienerFilter}
    H = \sum_{k=0}^{L-1} h_k A_e^k, 
\end{equation}
where the filter coefficients $h$ are found by solving
\begin{equation}
    \label{eq:WienerFilterCoefficients}
    \min_h \norm{B h - x}_2^2,\quad\text{with }B = [y \; A_e y \ldots
    \; A_e^{L-1} y].
\end{equation}
Using $R_{y,y}(\ell,m) = y^H (A_e^\ell)^H A_e^m y$ as definition for autocorrelation of the graph signal $y$, and $r_{x,y}(\ell) = y^H (A_e^\ell)^H x$
as definition of the cross-correlation between the graph signals $x$, yields the linear equation
\begin{equation}
    \label{eq:WienerFilterCoefficientsEquation}
    R_{y,y} h = r_{x,y}
\end{equation}
for the coefficients of the Wiener filter. Note that the powers of
$A_e$ can be computed efficiently using $A_e^k = V \Lambda_e^k V^{-1}$.

\mypar{Small graph signal} Since in~\cite{Gavili.Zahng:2017a} a random graph
was used to evaluate the graph Wiener filter, we use the USA graph in Fig.~\ref{fig:fixedUSAGraph} for our experiments. As graph signal we used, similar
to~\cite{Shafipour.Khodabakhsh.Mateos.Nikolova:2019a,Furutani.Shibahara.Akiyama.Hato.Aida:2019a},
the average monthly temperature of each state\footnote{Available on
  \url{https://www.currentresults.com/Weather/US/average-annual-state-temperatures.php}}. Then we added, over 1000 simulations, normally distributed noise with
zero mean and standard deviation of $10$ to the signal, leading to a
signal-to-noise ratio of $14.4 \pm 0.9$ decibel.

\mypar{Small graph results} Fig.~\ref{fig:NormalNoise10WienerFilter} shows the relative reconstruction error, as function of the filter order, for the graph Wiener filtered signal. The filter was designed with our modified graph that ensures diagonalizability. The qualitative behavior is as expected based on the results in~\cite{Gavili.Zahng:2017a}. Designing the Wiener filter based on the original graph and its Jordan basis fails (and was also not proposed in~\cite{Gavili.Zahng:2017a}). 

\begin{figure}
    \centering
    \includegraphics[width=0.8\linewidth]{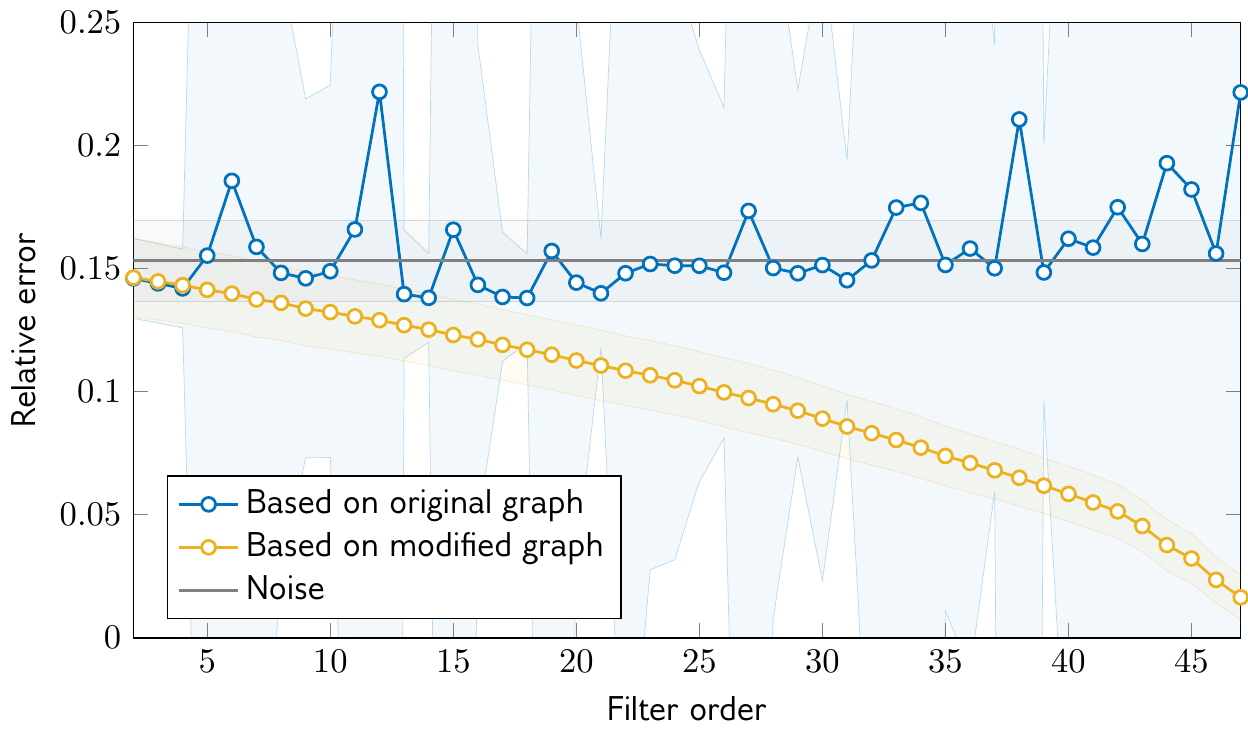}
    \caption{The relative reconstruction error $\norm{x - d}/\norm{x}$
      of the Wiener filtered noisy signal $d$ compared the original signal
      $x$, for different filter orders and energy-preserving shifts
      based on the adjacency matrix of the original graph and our
      modified graph. We used normally distributed noise with mean 0 and
      standard deviation 10. The average over 1000 noise simulations
      is shown as thick line and the standard deviation as shaded area
      in the respective colors.}
    \label{fig:NormalNoise10WienerFilter}
\end{figure}

\mypar{Large graph results} The experiments in~\cite{Gavili.Zahng:2017a} only considered a graph with 40 nodes. Here, we repeat the previous experiment with the large scale Manhattan graph signal in Fig.~\ref{fig:ManhattanGraphSignal} with added noise. Since the JNF is (by far) not computable in this case (and the method also did not work with the Jordan basis in Fig.~\ref{fig:NormalNoise10WienerFilter}) we only show the results for the modified graph after applying DestroyZeroEigenvalues and DestroyJordanBlocks and get a roughly similar behavior as before. Due to the high computational cost we show only one run and thus no standard deviation.

\begin{figure}
    \centering
    \includegraphics[width=0.8\linewidth]{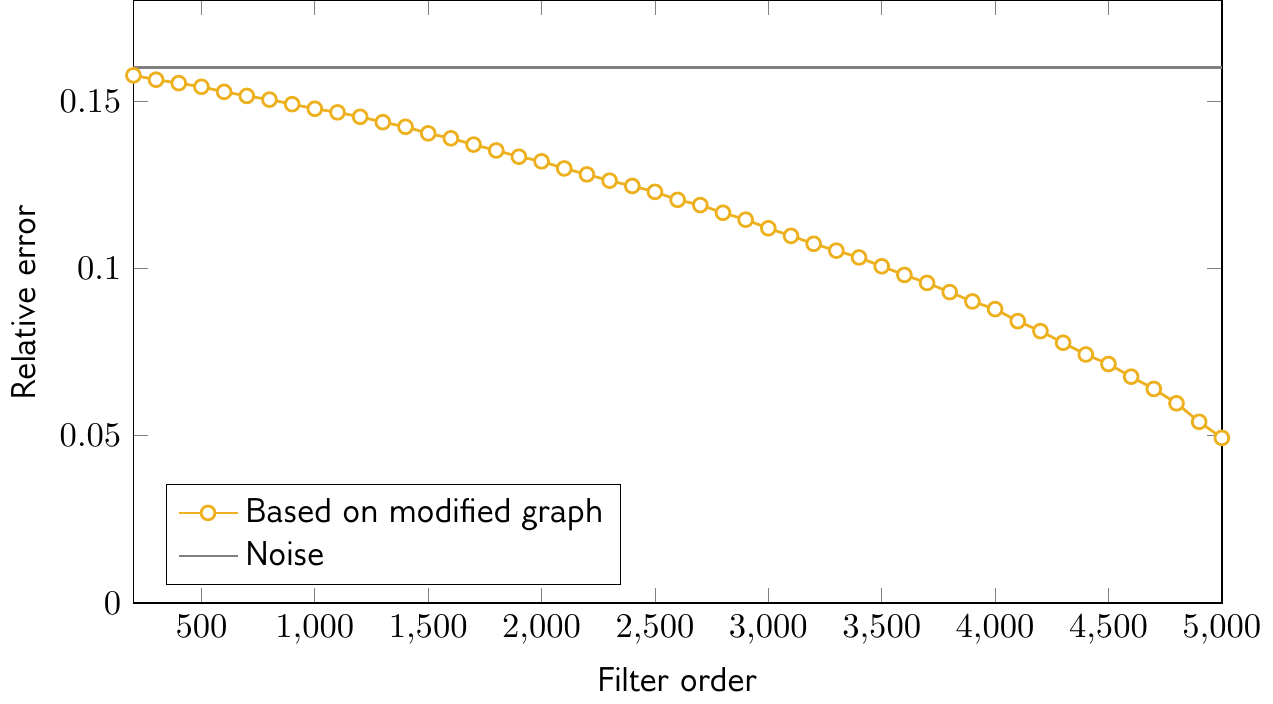}    
    \caption{The relative reconstruction error $\norm{x - d}/\norm{x}$
      of the Wiener filtered noisy signal $d$ to the original
      Manhattan graph signal $x$. In this experiment we added normally
      distributed noise with mean 0 and standard deviation 0.5.}
    \label{fig:ManhattanGraphWienerFilterError}
\end{figure}

In summary, using our method as preprocessing step makes the design of Wiener filters from~\cite{Gavili.Zahng:2017a} applicable to non-diagonalizable digraphs.

\subsection{Laplacian}

In this section we repeat some of the above experiments for the Laplacian instead of the adjacency
matrix, to demonstrate that our algorithm is equally applicable with minor modifications. For $D$ we use the in-degrees in all experiments. We have to slightly adjust DestroyJordanBlocks in
Fig.~\ref{algo:FixDigraphSpectrum}, replacing the condition $u_i \cdot v_j \not= 0$ with
$v_j (u_j - u_i) \not= 0$. Note that there is no analogue of the
algorithm from Fig.~\ref{algo:KillAllZeros} for the Laplacian, as no particular 
eigenvalue has a distinguished significance and zero is always an eigenvalue and thus cannot be destroyed. For example, in directed acyclic graphs Jordan blocks can only appear if there are nodes with the same number of incoming edges.

We also remind the reader that diagonalizability of $A$ and $L= D-A$ are different properties (see Fig.~\ref{fig:CounterexamplesLaplacianDiagAdjacency}) in general. 

\mypar{Random graphs} As for the experiment with the adjacency matrix, we generate 100
random digraphs for each of the random digraph models. The Laplacians
of the Erdős–Rényi random graphs are again all diagonalizable.
For the other three types of random digraphs the results
are shown in Table~\ref{tab:ResultsRandomGraphsLaplacian}. The overall
behavior is similar as before in Table~\ref{tab:ResultsRandomGraphs}.
\begin{table}
    \centering
    \begin{tabular}{@{}lllllllll@{}}\toprule
      & \multicolumn{2}{c}{min}& &
        \multicolumn{2}{c}{median}& &
        \multicolumn{2}{c}{max} \\
      \cmidrule{2-3} \cmidrule{5-6} \cmidrule{8-9}
      & edges & time & & edges & time & & edges & time \\
      \midrule
      Watts-Strogatz & 0 & 0.1s & & 0 & 0.12s & & 3 & 0.61s \\
      Barabási-Albert & 27 & 3.3s & & 40 & 4.8s & & 55 & 6.9s \\
      Klemm-Eguílez & 3 & 0.4s & & 12 & 1.6s & & 52 & 5.5s \\
      \bottomrule          
    \end{tabular}
    \caption{Edges added and runtime of DestroyJordanBlocks for three different
      random graph models with 500 nodes and approximately 5000
      edges.}
    \label{tab:ResultsRandomGraphsLaplacian}
\end{table}

\mypar{USA graph} The Jordan structure for the Laplacian of the USA
graph is shown in Table~\ref{tab:JUSA}. To obtain a diagonalizable
Laplacian on the USA graph our algorithm adds $6$ additional edges (one more than the theoretical minimum of 5), shown in Fig.~\ref{fig:fixedUSAGraphLaplacian}. The eigenvalues before
and after adding the edges are shown in Fig.~\ref{fig:USAGraphEigenvaluesLaplacian}.

\begin{table}\centering
\begin{tabular}{ll}\toprule
Eigenvalue & Block sizes\\ \midrule
0 & 1,1,1,1\\
1 & 4,3,2,1,1,1\\
2 & 4,2,1,1,1\\
3 & 8,4,3,2,2\\
4 & 2\\
5 & 2\\ \bottomrule
\end{tabular}
\caption{Jordan block sizes for Laplacian of the USA graph.\label{tab:JUSA}}
\end{table}

\begin{figure} \centering
    \includegraphics[width=0.8\linewidth]{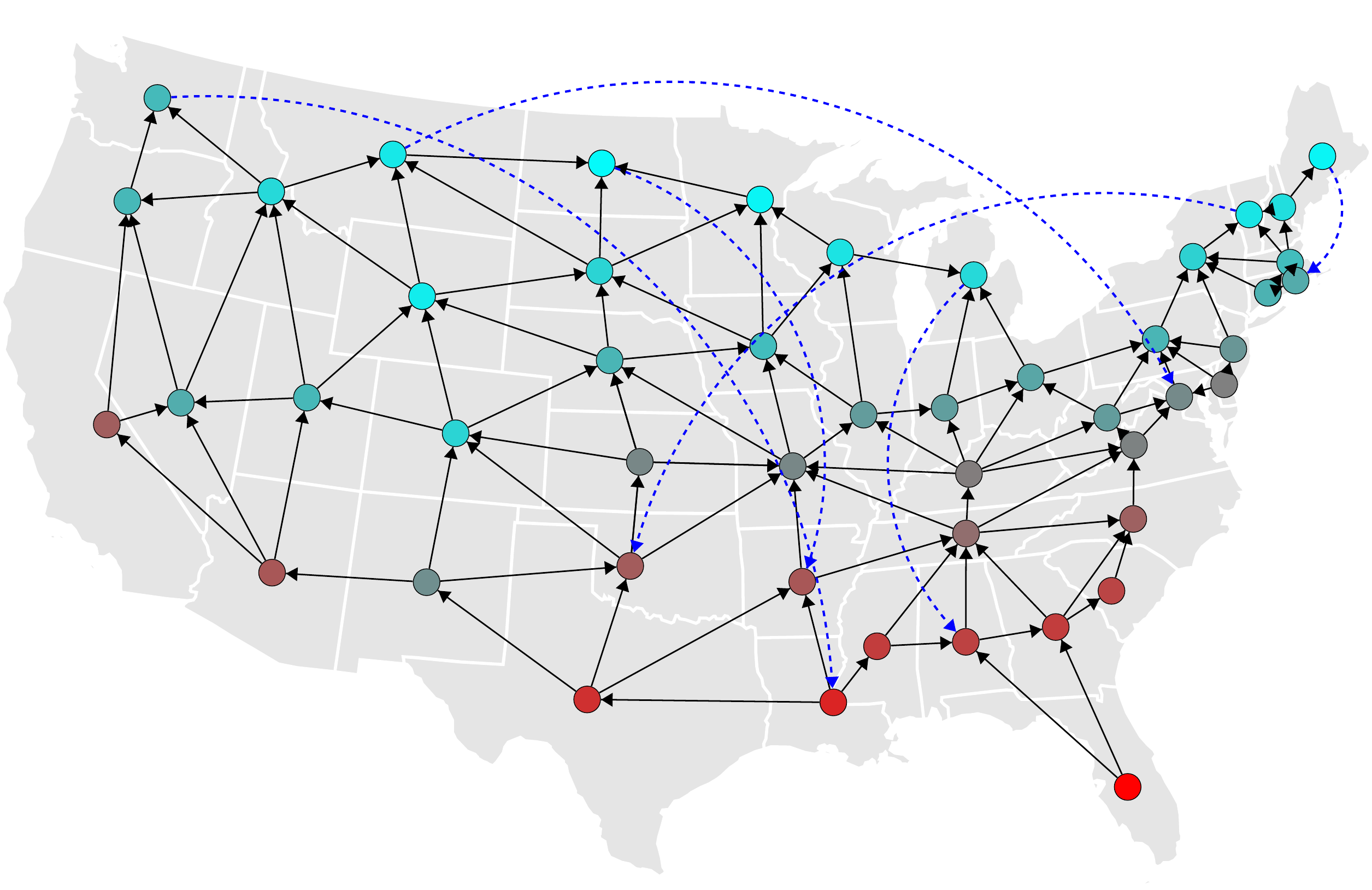}
    \caption{The USA graph. The 6 new edges added by
      DestroyJordanBlocks for the Laplacian are shown as dashed blue.}
    \label{fig:fixedUSAGraphLaplacian}
\end{figure}
\begin{figure}
    \centering
    \includegraphics[width=0.7\linewidth]{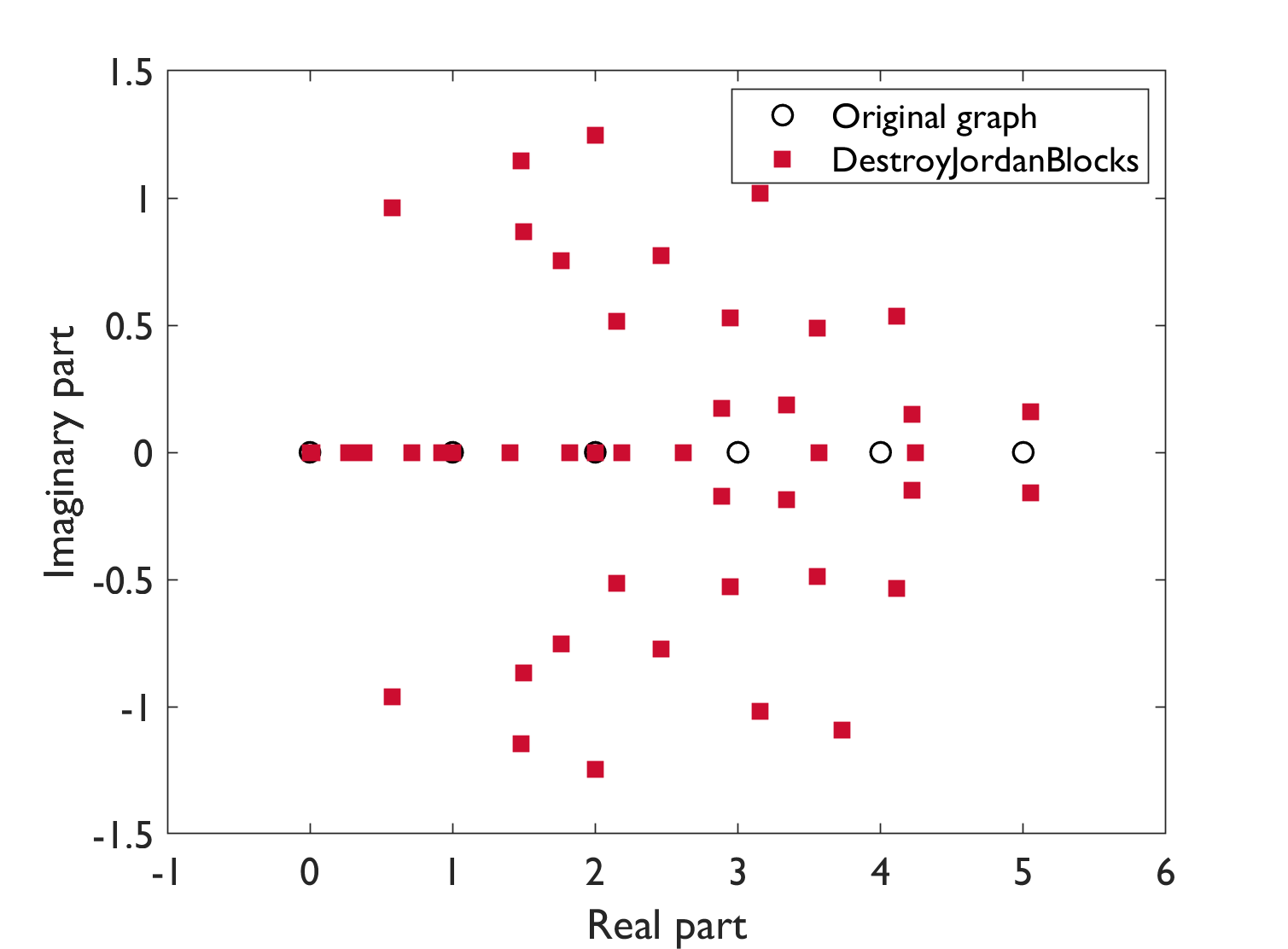}
    \caption{The Laplacian eigenvalues of the USA graph (black circles),
      and after making it diagonalizable (red squares).}
    \label{fig:USAGraphEigenvaluesLaplacian}
\end{figure}

\mypar{Manhattan taxi graph} We apply DestroyJordanBlocks to the
Laplacian of the Manhattan graph. The algorithm added $596$ new
edges within $12$ hours. The angles between the computed eigenvectors
before and after adding the edges are shown in
Fig.~\ref{fig:ManhattanGraphHistogramLaplacian}, i.e., the Fourier basis is close to orthogonal.
The eigenvalues $L + B$ are in the same range as those of $L$ (Fig.~\ref{fig:LaplacianSpectrumManhattan}).

\begin{figure}
    \centering
    \includegraphics[width=0.49\linewidth]{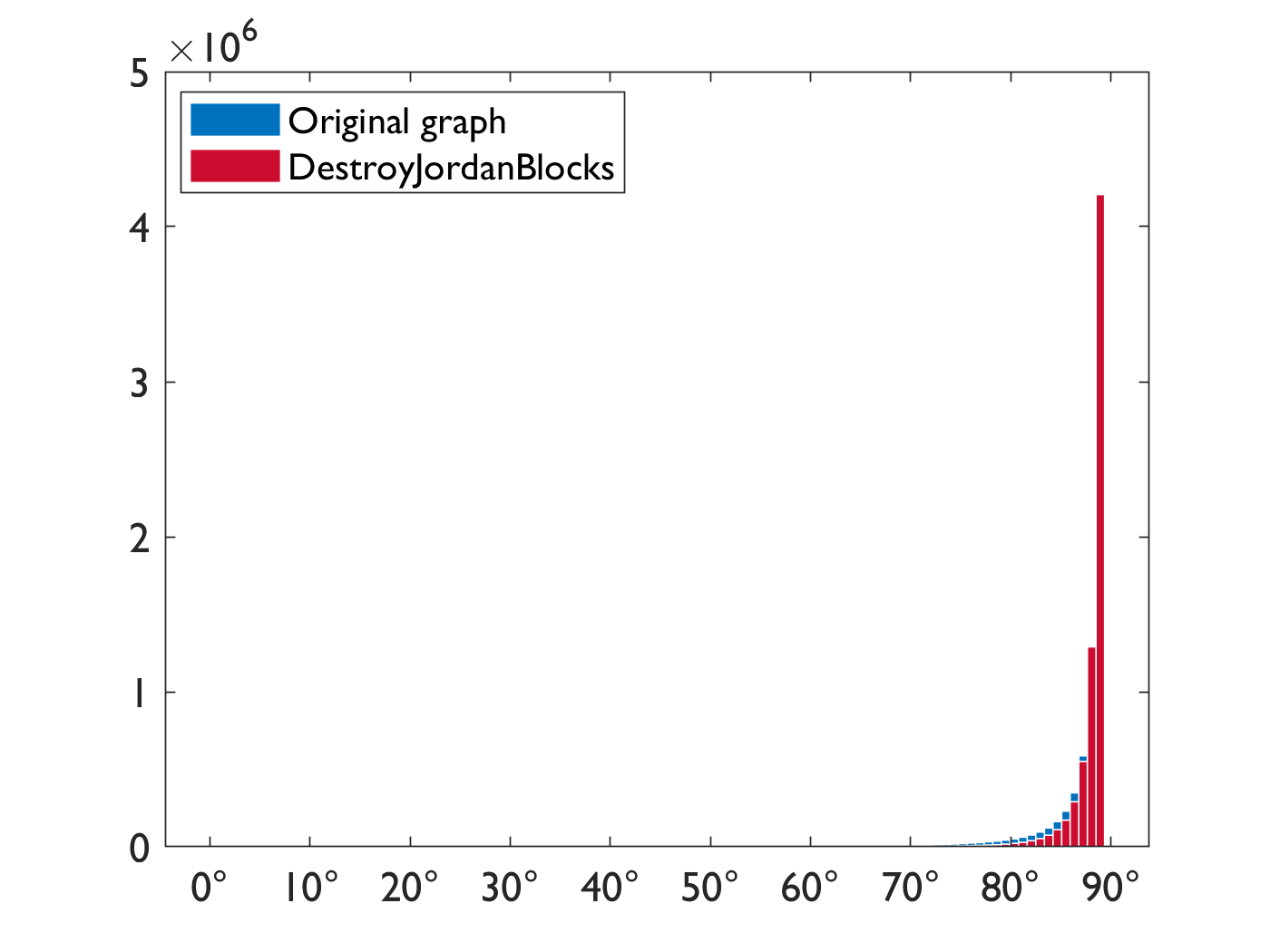}
    \includegraphics[width=0.49\linewidth]{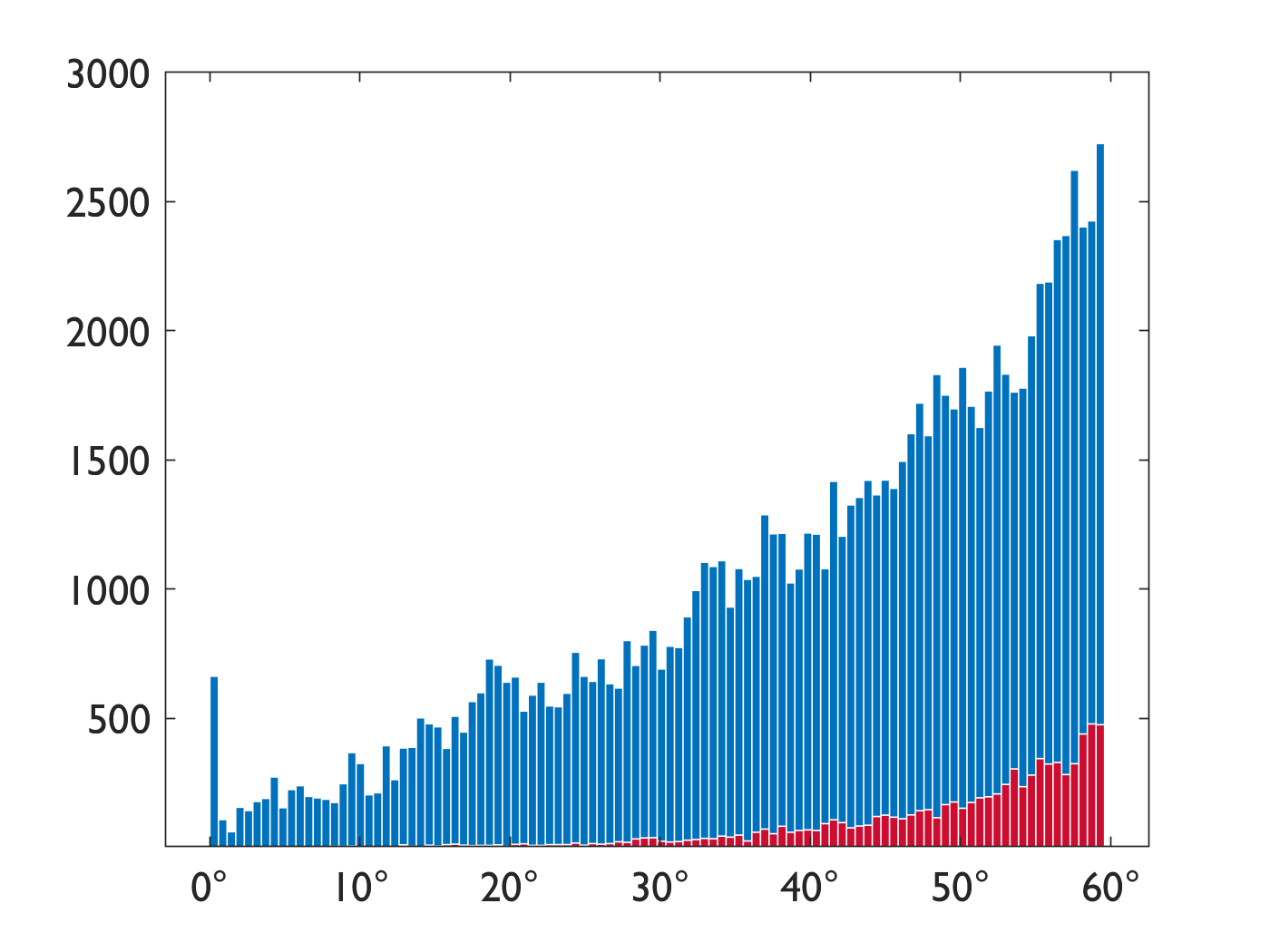}
    \caption{Manhattan graph: The histogram of all $5464^2$ angles
      between the spaces spanned by the computed eigenvectors of the
      Laplacian: for all angles (left), and zoomed in on angles
      $\leq 60$ degrees (right).}
    \label{fig:ManhattanGraphHistogramLaplacian}
\end{figure}
\begin{figure}
    \centering
    \includegraphics[width=\linewidth]{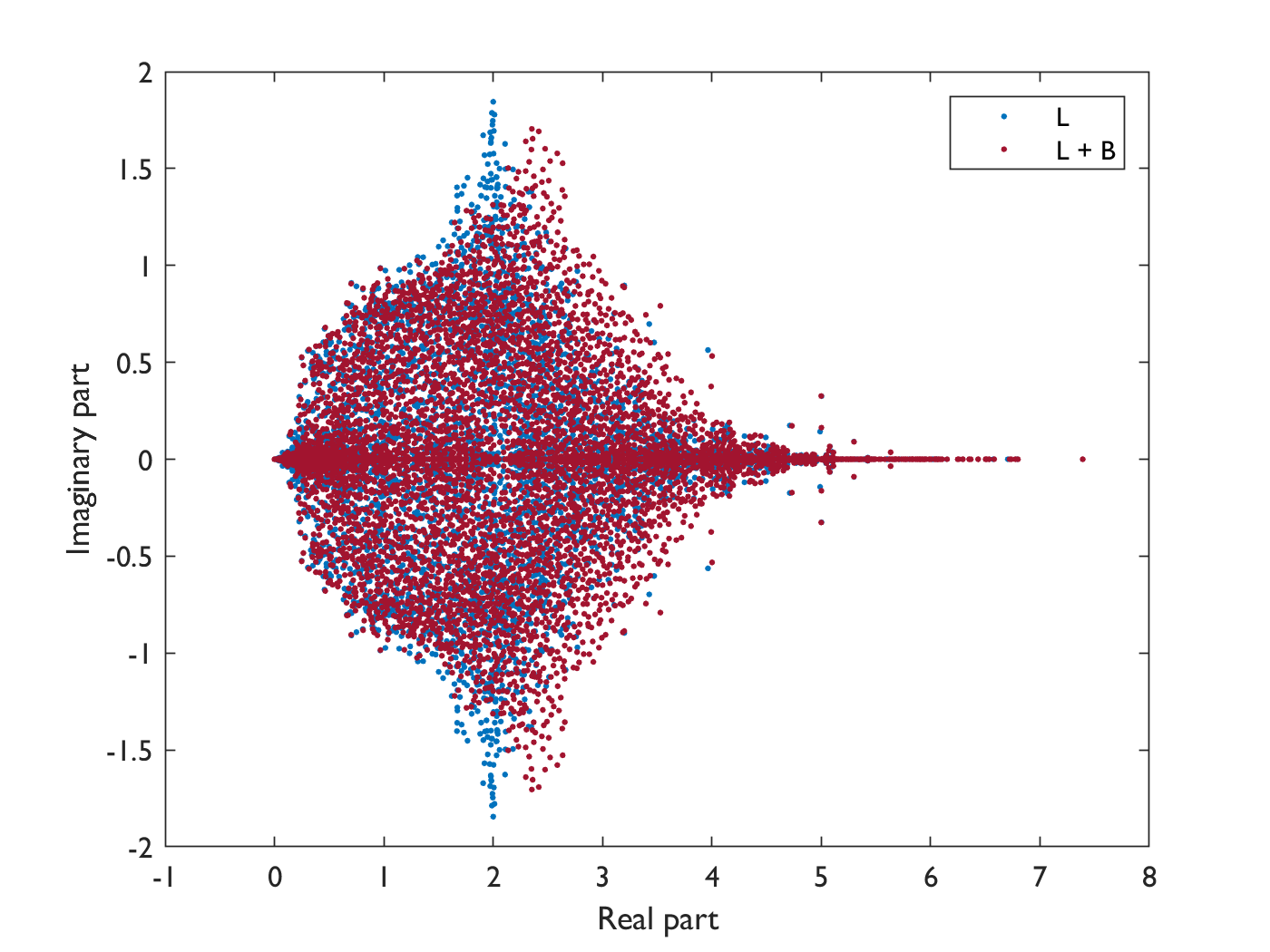}    
    \caption{Eigenvalues of the modified Laplacian of the Manhattan
      graph $L + B$.}
    \label{fig:LaplacianSpectrumManhattan}
\end{figure}


\section{Conclusion}
\label{sec:Conclusion}%

We presented a practical and scalable solution to the challenging
problem of designing a suitable Fourier basis in the case of non-diagonalizable shifts and filters in digraph signal processing. The basic idea was to add edges, i.e., slightly perturb the adjacency or Laplacian matrix to enforce this property. Then the Fourier basis and transform of the modified graph are used for the original graph.
Equivalently, our method can be seen as a way to construct an approximate, numerically stable eigenbasis and associated approximate Fourier transform that are still associated with an intuitive notion of shift in the graph domain. We showed that the method even works for directed acyclic graphs, which only have the eigenvalue zero.

Our method has more general potential uses to establish other desirable properties. Examples that we showed in the paper include invertibility or simple eigenvalues only for the graph shift, properties that are required or desirable for certain applications. It is intriguing, and invites further investigation, that the added edges must add cycles in the graph, thus generalizing the concept of cyclic boundary conditions. Also intriguing is that the Fourier bases obtained seem to be close to orthogonal and that they seem to maintain the total variation and its ordering with respect to the original graph. Finally, we would like to stress that the implementation of our method copes well with the inherent numerical instability of eigenvalue computations and scales to several thousands nodes.


\section*{Acknowledgments} 

The authors are very grateful to Sergey~V.~Savchenko (Landau Institute
for Theoretical Physics, Russian Academy of Sciences) for informing
them about an error in the formulation of
Theorem~\ref{thm:DestroyingJordanBlocks} in a previous version of the
paper and for providing us with Lemma~\ref{lemma:SoundHeuristics} and
its proof \cite{Savchenko:2020a}.

The authors also thank the anonymous reviewers for helpful comments which improved the content and the presentation of the paper.



\bibliographystyle{IEEEbib}
\bibliography{Literature/AlgebraicSignalProcessing,Literature/DataSources,Literature/GraphSignalProcessing,Literature/GraphTheory,Literature/LinearAlgebra,Literature/Misc,Literature/PerturbationTheory}
\begin{IEEEbiography}[{\includegraphics[width=1in,height=1.25in,clip,keepaspectratio]{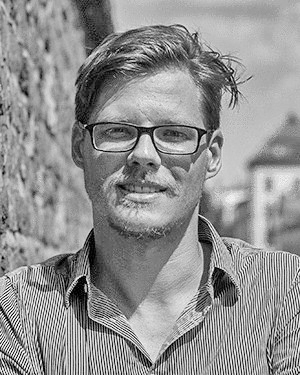}}]{Bastian  Seifert}
    (Member, IEEE) received the B.Sc. degree in mathematics from Friedrich-Alexander-Universität Erlangen-Nürnberg, Germany, in 2013, and the M.Sc. and Ph.D. degrees in mathematics from Julius-Maximilians-Universität Würzburg, Germany, in 2015 and 2020, respectively. He was a Research Associate with the Center for Signal Analysis of Complex Systems (CCS), the University of Applied Sciences, Ansbach, Germany. Currently, he is a Postdoc at ETH Zurich, Switzerland. His research interests include algebraic signal processing, dimensionality reduction, and applied mathematics.
\end{IEEEbiography}

\begin{IEEEbiography}[{\includegraphics[width=1in,height=1.25in,clip,keepaspectratio]{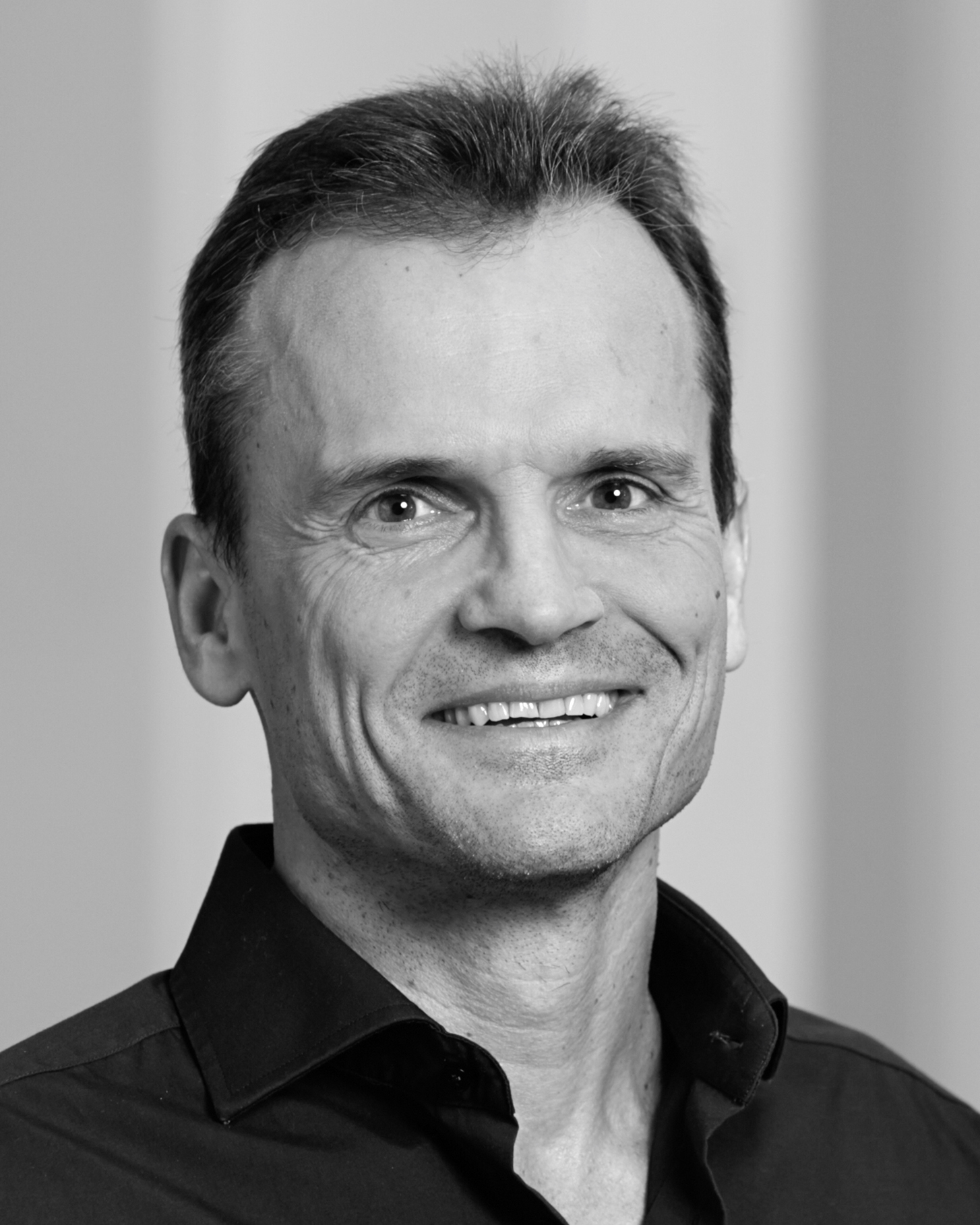}}]{Markus P\"uschel}
	(Fellow, IEEE) received the Diploma (M.Sc.) in mathematics and Doctorate 1068
	(Ph.D.) in computer science, in 1995 and 1998, respectively, both from the University of Karlsruhe, Germany. He is a Professor of Computer Science with ETH Zurich, Switzerland, where he was the Head of the Department from 2013 to 2016. Before joining ETH in 2010, he was a Professor with Electrical and Computer Engineering, Carnegie Mellon University (CMU), where he still has an Adjunct status. He was an Associate Editor for the IEEE Transactions on Signal Processing, the IEEE Signal Processing Letters, and was a Guest Editor of the Proceedings of the IEEE and the Journal of Symbolic Computation, and served on numerous program committees of conferences in computing, compilers, and programming languages. He received the main teaching awards from student organizations of both institutions CMU and ETH and a number of awards for his research. His current research interests include algebraic signal processing, program generation, program analysis, fast computing, and machine learning.
\end{IEEEbiography}

\vfill\pagebreak




\end{document}